\documentclass[a4paper,8pt]{article}
\usepackage{soul}
\usepackage[usenames,dvipsnames]{color}
\usepackage{jheppub}
\usepackage{esvect}
\usepackage{relsize}
\usepackage{comment}
\usepackage{amsmath, amssymb, slashed, epsf, color, graphicx, latexsym, bm}
\usepackage{mathrsfs}
\usepackage{tensor}
\usepackage{fixme}
\fxsetup{status=draft}
\usepackage{physics}
\usepackage{tikz}
\usetikzlibrary{positioning, calc}
\usepackage{dsfont}
\usepackage{epsfig}
\usepackage{graphics}
\usepackage[T1]{fontenc}
\usepackage{mathtools}
\usepackage{fixme}
\fxsetup{status=draft}
\usepackage{caption}
\usepackage{float}
\usepackage{bbm}
\usepackage{subcaption}
\usepackage{enumitem}
\usepackage{verbatim}
\usepackage{blindtext}
\makeatletter
\renewcommand*\env@matrix[1][*\c@MaxMatrixCols c]{%
  \hskip -\arraycolsep
  \let\@ifnextchar\new@ifnextchar
  \array{#1}}
\makeatother
\DeclareFontEncoding{LS1}{}{}
\DeclareFontSubstitution{LS1}{stix}{m}{n}
\DeclareSymbolFont{symbols2}{LS1}{stixfrak}{m}{n}
\DeclareMathSymbol{\typecolon}{\mathbin}{symbols2}{"25}
\usepackage[mathlines]{lineno}
\newcommand{\myrightleftarrows}[1]{\mathrel{\substack{\xrightarrow{#1} \\[-.9ex] \xleftarrow{#1}}}}
\newcommand{\C}{\mathbb{C}}
\newcommand{\1}{\mathds{1}}

\newcommand{\N}{\mathbb{N}}
\newcommand{\Z}{\mathbb{Z}}

\newcommand{\F}{\mathbb{F}}
\newcommand{\s}[1]{\mathcal{#1}}
\newcommand{\mr}[1]{\mathrm{#1}}

\newcommand{\wt}[1]{\widetilde{#1}}
\newcommand{\wh}[1]{\widehat{#1}}

\begin{document}
\title{Algebraic Structures In Closed Superstring Field Theory, Homotopy Transfer And Effective Actions}
\author[a]{Ranveer Kumar Singh} 
\affiliation[a]{NHETC and Department of Physics and Astronomy, Rutgers University, 126
Frelinghuysen Rd., Piscataway NJ 08855, USA}
\emailAdd{ranveersfl@gmail.com}
\abstract{A consistent action for heterotic and type II superstring field theory was recently proposed by Sen. We give an algebraic formulation of this action in terms of certain twisted $L_\infty$-algebra. We further show that Sen's Wilsonian effective superstring field action can be obtained using homotopy transfer and the effective theory also possesses the algebraic structure of a twisted $L_\infty$-algebra.} 
\maketitle
\section{Introduction and summary}
String field theory is an attempt to formulate a non-perturbative description of string theory. Consistent construction of string field theories exists for all string theories: bosonic open and closed strings \cite{Siegel:1984ogw,Kaku1974FieldTO,Witten:1985cc,Zwiebach:1992ie,Zwiebach:1997fe}, heterotic strings and type II strings \cite{Sen:2015uaa,Sen:2015hha,Sen:2014dqa} and open-closed superstrings \cite{FarooghMoosavian:2019yke}, see reviews \cite{Erbin:2021smf,Erler:2019loq,deLacroix:2017lif,Sen:2024nfd} for a comprehensive introduction\footnote{See also \cite{Maccaferri:2023vns} for a quick exposure to the recent results in string field theory}. Most of these string field theories are succinctly described in terms of algebraic structures called \textit{homotopy algebras} first introduced in mathematics by J. Stasheff \cite{58e41a3e-2db4-3e2f-8edb-8effbf4ad0ce}. The two basic structures describing open and closed string field theories are $A_\infty$ and $L_\infty$-algebras \cite{Zwiebach:1992ie,Zwiebach:1997fe,Gaberdiel:1997ia,Kajiura:2003ax,Kajiura:2004xu,Markl:1997bj,Kajiura:2006mt,Kajiura:2001ng}. The basic structure for these algebras is the existence of $n$-multilinear maps, called $n$-ary products (or string products) and denoted by $\{m_n\}_{n\geq 1}$ for $A_\infty$ algebras and $\{\ell_n\}_{n\geq 1}$ for $L_\infty$-algebras,  which satisfy a \textit{strong homotopy Jacobi identity} which generalises associativity. String vertices in closed and open string theories satisfy similar identity, called the \textit{main identity} in string field theory literature. These string vertices can then be used to define \textit{string products} which then give rise to homotopy algebras. \par Formulation of string field theory opens up possibilities to apply the methods of quantum field theory to resolve various problems in perturbative string theory. In particular, string field theory has been used to regularise IR divergences in string theory and compute D-instanton amplitudes \cite{Sen:2014pia,Sen:2015hia,Sen:2016qap,Sen:2019qqg,Sen:2020eck,Sen:2021jbr,Sen:2021qdk,Sen:2021tpp,Alexandrov:2021shf}. Another important application is to write down Wilsonian effective superstring actions by integrating out heavy modes to get an effective action for light fields. By integrating out massive modes in open (super)string field theory, one can extract the (super)Yang-Mills action \cite{Berkovits:2003ny,Asada:2017ykq}. It turns out that effective actions can be understood in terms of \textit{homotopy transfer} of homotopy algebras. For string field theories based on $A_\infty$ and $L_\infty$-algebras, a beautiful demonstration of this procedure was given in \cite{Erbin:2020eyc} (see also \cite{Arvanitakis:2020rrk,Arvanitakis:2021ecw}). \footnote{See \cite{Maccaferri:2021lau,Maccaferri:2023gcg,Maccaferri:2023gof,doubek2019quantum,Jurco:2018sby,Jurco:2020yyu,Macrelli:2019afx,Borsten:2021hua} for various applications of homotopy transfer in string field theory.}. 
A broad brush sketch of this story goes as follows: Let $\s{H}$ be the (super)string Hilbert space.
\begin{enumerate}
    \item One begins by formulating the $A_\infty$ ($L_\infty$)-algebra underlying the theory in terms of tensor colagebra $T\s{H}=\bigoplus_{n\geq 0}\s{H}^{\otimes n}$ (symmetrized tensor coalgebra $S\s{H}=\bigoplus_{n\geq 0}\s{H}^{\wedge n}$), see Appendix \ref{app:coalg} for definitions. The string products $m_n:\s{H}^{\otimes n}\to\s{H}$ ($\ell_n:\s{H}^{\wedge n}\to\s{H}$) on $\s{H}$ then extends \footnote{See \eqref{eq:coderbmci} for the construction of this extension to $S\s{H}$.} to coderivations $\bm{m}_n:T\s{H}\to T\s{H}$ ($\bm{\ell}_n:S\s{H}\to S\s{H}$). The strong homotopy Jacobi identity of the $n$-ary product is then equivalent to the \textit{total $n$-ary product} $\bm{m}=\sum_{n\geq 1}\bm{m}_n$ ($\bm{\ell}=\sum_{n\geq 1}\bm{\ell}_n$) being nilpotent: $\bm{m}^2=0$ ($\bm{\ell}^2=0$). 
    \item The free theory is described by a (symmetrized) tensor coalgebra with only the 1-ary product, which is nilpotent. In string field theory, this nilpotent operator is the BRST charge $m_1:=Q_B$ ($\ell_1:=Q_B$). 
    \item Suppose $P$ is the projection operator such that the Wilsonian effective action is obtained by integrating out states not in the image of $P$. Then the Wilsonian effective free theory is trivially obtained by mapping the BRST charge to its compression \footnote{The compression of a linear operator $T$ on a Hilbert space $H$ to a subspace $K$ is the operator
$P_K T|_K: K \rightarrow K,$
where $P_K: H \rightarrow K$ is the orthogonal projection onto $K$.} $PQ_B|_{P\s{H}}$. The free theory and the free effective theory can be described by a \textit{strong deformation retract} (SDR) $(T\s{H},\bm{Q_B})\to (TP\s{H},\mathbf{P}\bm{Q_B}|_{TP\s{H}})$ in case of $A_\infty$-algebra and $(S\s{H},\bm{Q_B})\to (SP\s{H},\mathbf{P}\bm{Q_B}|_{SP\s{H}})$ in case of $L_\infty$-algebra. An important axiom for an SDR is the nilpotency of the operators: 
    \begin{equation}
        \begin{split}
          &\bm{Q_B}^2=0,\quad (\mathbf{P}\bm{Q_B}|_{TP\s{H}})^2=0~,\\&\bm{Q_B}^2=0,\quad (\mathbf{P}\bm{Q_B}|_{SP\s{H}})^2=0~.
        \end{split}
    \end{equation}
    \item The interacting theory is obtained from the free theory by introducing a perturbation $\delta\bm{m}=\sum_{n\geq 2}\bm{m}_n$ for $A_\infty$-algebra and $\delta\bm{\ell}=\sum_{n\geq 2}\bm{\ell}_n$ for $L_\infty$-algebra. The main step is the application of the \textit{homological perturbation lemma} which gives a new SDR if the perturbation satisfies certain conditions, see \cite{Erbin:2020eyc} for the precise statement. Thus we obtain a new SDR $(T\s{H},\bm{Q_B}+\delta\bm{m})\to (TP\s{H},\mathbf{P}\bm{Q_B}|_{TP\s{H}}+\delta\wt{\bm{m}})$ in case of $A_\infty$-algebra and $(S\s{H},\bm{Q_B}+\delta\bm{\ell})\to (SP\s{H},\mathbf{P}\bm{Q_B}|_{SP\s{H}}+\delta\wt{\bm{\ell}})$ in case of $L_\infty$-algebra. The $n$-ary products $\{\wt{m}_n\}_{n\geq 1},\{\wt{\ell}_n\}_{n\geq 1}$ for the effective theory is given by the expansion
    \begin{equation}
     \begin{split}
         &\mathbf{P}\bm{Q_B}|_{TP\s{H}}+\delta\wt{\bm{m}}=\sum_{n\geq 1}\wt{\bm{m}}_n~,\\&\mathbf{P}\bm{Q_B}|_{SP\s{H}}+\delta\wt{\bm{\ell}}=\sum_{n\geq 1}\wt{\bm{\ell}}_n.
     \end{split}   
    \end{equation}
They satisfy the strong homotopy Jacobi identity by the axioms of an SDR. The effective action can then be written using the effective string products. This procedure of constructing new homotopy algebra from old is called homotopy transfer.
\end{enumerate}
The aim of this paper is to demonstrate the same for heterotic and type II superstring field theory. The procedure of integrating out heavy modes to obtain Wilsonian effective action for light modes in superstring field theory was given in \cite{Sen:2016qap}. 
\par 
Algebraic construction of various superstring field theories based on homotopy algebras appeared in \cite{Kunitomo:2021wiz,Kunitomo:2022qqp,Kunitomo:2019glq,Erler:2017onq,Konopka:2016grr,Erler:2016ybs,Kunitomo:2015usa} and a formulation of Sen's heterotic string field action in terms of standard $L_\infty$-algebra appeared in \cite{Erler:2017pgf} (see also \cite{Erler:2019loq}). It turns out that the formulations in some of these references (for example \cite{Erler:2017pgf}), though capture Sen's formulation of the heterotic superstring action, they are not suitable for the application of homotopy transfer to obtain Sen's Wilsonian effective actions. 
\par In this paper, we define the notion of a $g$-twisted $L_\infty$-algebra which give a uniform algebraic formulation of Sen's heterotic and type II string field action.  A $g$-twisted $L_\infty$-algebra consists of the following data:
\begin{enumerate}
    \item A $\Z$-graded vector space $V$ with dual denoted by $V^\star$.
    \item A linear map $g:V^\star\longrightarrow V$.
    \item A family of maps $\{\ell_n\}_{n=1}^\infty$ where $\ell_1:V\oplus V^\star\to V\oplus V^\star$ is nilpotent and $\ell_n:V^{\wedge n}\to V^\star$ is multilinear. Here $V^{\wedge n}$ is the degree $n$ symmetrized tensor product of $V$, see \eqref{eq:symm_tensor_vn} for the precise definition. 
\end{enumerate}
This data satisfies the strong homotopy Jacobi identity, see Definition \ref{def:g_twisted_L_infty} for details. In particular, $\ell_1$ is nilpotent. \par The  
description of heterotic and type II superstring theory given in \cite{Sen:2015uaa,deLacroix:2017lif,FarooghMoosavian:2019yke} is naturally captured by a $\s{G}$-twisted $L_\infty$-algebra where $\s{G}$ is given in terms of the zero-mode of the picture changing operator. The vector space $V$ is taken to be the subspace of the string Hilbert space consisting of states with picture number
\begin{equation}
\begin{split}
    -1,-\frac{1}{2}\quad&\text{in heterotic}~,\\
    (-1,-1),(-\frac{1}{2},-1),(-1,-\frac{1}{2}),(-\frac{1}{2},-\frac{1}{2})\quad &\text{in type II}~,
\end{split}    
\end{equation}
see Section \ref{sec:2.1} for details. The BPZ inner product then identifies the dual $V^\star$ this subspace with another subspace consisting of states with picture number 
\begin{equation}
\begin{split}
    -1,-\frac{3}{2}\quad&\text{in heterotic}~,\\
    (-1,-1),(-\frac{3}{2},-1),(-1,-\frac{3}{2}),(-\frac{3}{2},-\frac{3}{2})\quad &\text{in type II}~.
\end{split}    
\end{equation}
The zero-mode $\s{G}$ of the picture changing operator is then the linear map $\s{G}:V^\star\to V$. The BRST charge $Q_B$ is the linear map $\ell_1$ and the $n$-string vertices (see Section \ref{sec:2.2}) are used to construct the multilinear maps $\ell_n$. The main identity (see \eqref{eq:mainid{}}) for the string vertices is then equivalent to the strong homotopy Jacobi identity for $\ell_n$. The details of the proof are presented in Section \ref{sec:2.4}.
\par 
We recast the discussion of Wilsonian effective action in terms of an SDR and then apply a slightly modified version of the homological perturbation lemma, appropriate for application to the present case, to show that Sen's Wilsonian effective action in \cite{Sen:2016qap} can be obtained using homotopy transfer. 
\par
We close this section by listing two potential future directions. 
\begin{enumerate}
\item The present formulation of Wilsonian effective action in terms of homotopy transfer is only for the classical string field action, i.e., the string vertices are calculated using only genus zero correlators of the string worldsheet. The quantum superstring field action includes all higher genus contributions to the string vertex and the main identity \eqref{eq:mainid{}} also receives contributions from integrals over other patches of the moduli space. The formulation of the quantum superstring action will require appropriate generalisation of the twisted $L_\infty$-algebra to \textit{quantum twisted $L_\infty$-algebra} in analogy to the standard $L_\infty$ case \cite{Zwiebach:1992ie,Sen:2024nfd}. It would also be desirable to understand the quantum Wilsonian effective action in terms of homotopy transfer.   
\item A successful formulation of quantum superstring field theory action and quantum Wilsonian effective action in terms of a quantum twisted $L_\infty$-algebra opens up the possibility of computing higher-derivative corrections to low energy supergravity using algebraic methods. It would be extremely interesting if homotopy transfer could be applied to reproduce known higher-derivative corrections to low energy supergravity.
\end{enumerate}
The paper is organised as follows: in Section \ref{sec:2}, more specifically in Section \ref{sec:2.1} and Section \ref{sec:2.2} we review the basics of closed superstring field theory. In Section \ref{sec:genLinfalg} and Section \ref{sec:2.4}, we define the notion of $g$-twisted $L_\infty$-algebra and show that the closed superstring field theory can be described in terms of these homotopy algebras. In Section \ref{sec:3} we review Wilsonian effective theory and perform some preliminary calculations. Finally in Section \ref{sec:4} we show that the effective string products can be obtained using homotopy transfer. In Appendix \ref{app:coalg} we recall some basic notions about (symmetrised) tensor coalgebras and recast $L_\infty$-algebras in the language of symmetrised tensor coalgebras. In Appendix \ref{app:sdr} we review the definition of an SDR and state the version of homological perturbation lemma used in this paper to obtain effective string products.   
\section{Closed superstring field theory as an $L_\infty$-algebra}\label{sec:2}
\subsection{The Hilbert space of superstring theory}\label{sec:2.1}
Let us recall the basic ingredients of closed superstring theory. We follow \cite{deLacroix:2017lif,Sen:2015uaa} for this exposition.\\\\
Let $\s{H}_T$ be the Hilbert space of GSO even states in the small Hilbert space \cite{Friedan:1985ey,Friedan:1985ge} of the matter-ghost CFT with arbitrary ghost and picture number satisfying
\begin{equation}\label{eq:smhilsp}
b_0^{-}|s\rangle=0, \quad L_0^{-}|s\rangle=0, \quad|s\rangle \in \mathcal{H}_T~,
\end{equation}
where we have defined 
\begin{equation}
b_0^{ \pm}=b_0 \pm \bar{b}_0, \quad L_0^{ \pm}=L_0 \pm \bar{L}_0, \quad c_0^{ \pm}=\frac{1}{2}\left(c_0 \pm \bar{c}_0\right) .
\end{equation}
$\s{H}_T$ decomposes into NS and R sectors for the heterotic string and into four sectors for type II:
\begin{equation}
\begin{array}{ll}
\s{H}_T=\s{H}_{\mr{NS}} \oplus \s{H}_{\mr{R}}, &\quad \text { Heterotic }, \\
\s{H}_T=\s{H}_{\mr{NSNS}} \oplus \s{H}_{\mr{NSR}} \oplus \s{H}_{\mr{RNS}} \oplus \s{H}_{\mr{RR}}, &\quad \text { Type II }.
\end{array}
\end{equation}
For heterotic string, let $\s{H}_p$ be the subspace of $\s{H}_T$ with picture number $p$. For NS sector, $p \in \mathbb{Z}$ and for R sector $p \in \mathbb{Z}+\frac{1}{2}$. For type II, let $\mathcal{H}_{p, q}$ denote the subspace of $\s{H}_T$ with left moving picture number $p$ and right-moving picture number $q$. Define the spaces
\begin{equation}\label{eq:picnumHTTtilT}
\begin{split}
&\widehat{\s{H}}_T=\begin{cases}
\s{H}_{-1} \oplus \s{H}_{-\frac{1}{2}},&\quad\text{Heterotic},\\   \s{H}_{-1,-1} \oplus \s{H}_{-\frac{1}{2},-1} \oplus \mathcal{H}_{-1,-\frac{1}{2}} \oplus \mathcal{H}_{-\frac{1}{2},-\frac{1}{2}},&\quad\text{Type II}, 
\end{cases}\\&\widetilde{\s{H}}_T=\begin{cases}
\s{H}_{-1} \oplus \s{H}_{-\frac{3}{2}},&\quad\text{Heterotic},\\   \s{H}_{-1,-1} \oplus \mathcal{H}_{-\frac{3}{2},-1} \oplus \mathcal{H}_{-1,-\frac{3}{2}} \oplus \mathcal{H}_{-\frac{3}{2},-\frac{3}{2}}, &\quad\text{Type II}. 
\end{cases}
\end{split}
\end{equation}
We will denote the Hilbert space of the theory by 
\begin{equation}
    \s{H}_T:=\widetilde{\s{H}}_T\oplus\widehat{\s{H}}_T.
\end{equation}
Let $\left\{\left|\varphi_r\right\rangle\right\} \subset \widehat{\s{H}}_T,\left\{\left|\varphi_r^c\right\rangle\right\} \subset \widetilde{\s{H}}_T$ be a basis of $\widehat{\s{H}}_T$ and $\widetilde{\s{H}}_T$ satisfying
\begin{equation}\label{eq:inprodphirc}
\left\langle\varphi_r^c\left|c_0^{-}\right| \varphi_s\right\rangle=\delta_{r s}, \quad\left\langle\varphi_s\left|c_0^{-}\right| \varphi_r^c\right\rangle=\delta_{r s}~.
\end{equation}
Here $\langle\cdot|\cdot\rangle$ is the BPZ inner product and the $c_0^-$ insertion is needed for a nonzero inner product because of the anticommutator $\{b_0^-,c_0^-\}=1$ and the condition \eqref{eq:smhilsp}. 
The first relation in \eqref{eq:inprodphirc} implies the second relation in \eqref{eq:inprodphirc} and also a completeness relation
\begin{equation}
\sum_r\left|\varphi_r \rangle\langle \varphi_r^c\right| c_0^{-}=\mathds{1},\quad \sum_r\left|\varphi_r^c \rangle\langle \varphi_r\right| c_0^{-}=\mathds{1} .
\end{equation}
The basis states $\varphi_r, \varphi_r^c$ carry nontrivial grassmann parities which we shall denote by $(-1)^{\gamma_r}$ and $(-1)^{\gamma_r^c}$ respectively. In the NS sector of the heterotic string and NSNS and RR sector of type II, the grassmann parity of $\varphi_r$ and $\varphi_r^c$ is odd and even respectively if the ghost number of $\varphi_r$ and $\varphi_r^c$ is odd and even respectively. In the R sector of the heterotic string and NSR and RNS sector of type II, the grassmann parity of $\varphi_r$ and $\varphi_r^c$ is odd and even respectively if the ghost number of $\varphi_r$ and $\varphi_r^c$ is even and odd respectively. Moreover one has
\begin{equation}
(-1)^{\gamma_r+\gamma_r^c}=-1~.
\end{equation}
Introduce the zero modes of the PCOs:
\begin{equation}
\s{X}_0 \equiv \oint \frac{d z}{z} \s{X}(z), \quad \bar{\s{X}}_0=\oint \frac{d \bar{z}}{\bar{z}} \bar{\s{X}}(\bar{z}), 
\end{equation}
where the contour integral is normalised such that 
\begin{equation}
    \oint_C \frac{dz}{z}=1,\quad \oint_{\bar{C}} \frac{d\bar{z}}{\bar{z}}=1~,
\end{equation}
and $C,\Bar{C}$ is a contour around $z=0,\bar{z}=0$ respectively, see \cite{deLacroix:2017lif} for the definition of the operator $\s{X}$. 
Define the operator $\s{G}$ as
\begin{equation}
\s{G}= \begin{cases}\mathds{1} \oplus \s{X}_0, & \mathcal{H}_T=\mathcal{H}_{\mr{N S}} \oplus \mathcal{H}_{\mr{R}}~, \\ \mathds{1} \oplus \s{X}_0 \oplus \bar{\s{X}}_0 \oplus \s{X}_0 \bar{\s{X}}_0, & \s{H}_T=\s{H}_{\mr{NSNS}} \oplus \s{H}_{\mr{RNS}} \oplus \s{H}_{\mr{NSR}} \oplus \s{H}_{\mr{RR}}.\end{cases}
\end{equation}
on the heterotic and type II Hilbert space. One has that
\begin{equation}
\left[\s{G}, L_0^{ \pm}\right]=0,\quad\left[\s{G}, b_0^{ \pm}\right]=0,\quad\left[\s{G}, Q_B\right]=0 \text {. }
\end{equation}
It also satisfies
\begin{equation}\label{eq:Ghermbpz}
    \langle \wt{\Psi}|c_0^-\s{G}|\wt{\Psi}'\rangle=\langle \s{G}\wt{\Psi}|c_0^-|\wt{\Psi}'\rangle~,\quad |\wt{\Psi}\rangle,|\wt{\Psi}'\rangle\in \wt{\s{H}}_T.
\end{equation}
\subsection{Superstring field theory action}\label{sec:2.2}
A string field is a grassmann even element of $|\Psi\rangle \in \widehat{\s{H}}_T$. We also need another non-dynamical grassmann even string field $|\widetilde{\Psi}\rangle \in \widetilde{\s{H}}_T$ to write down the action. The  1PI effective action is written as
\begin{equation}
S_{1PI}=\frac{1}{g_s^2}\left[-\frac{1}{2}\left\langle\widetilde{\Psi}\left|c_0^{-} Q_B \s{G}\right| \widetilde{\Psi}\right\rangle+\left\langle\widetilde{\Psi}\left|c_0^{-} Q_B\right| \Psi\right\rangle+\sum_{n=1}^{\infty} \frac{1}{n !}\left\{\Psi^n\right\}\right]~,
\end{equation}
where $\left\{\Psi^n\right\}$ is the $n$-string vertex summed over all 1PI Feynman diagrams, see \cite{deLacroix:2017lif,Sen:2015uaa} for precise definition in terms of integral over moduli space of Riemann surfaces. 
This action has infinite dimensional gauge invariance given by 
\begin{equation}
|\delta \Psi\rangle=Q_B|\Lambda\rangle+\sum_{n=0}^{\infty} \frac{1}{n !} \mathcal{G}\left[\Psi^n \Lambda\right], \quad|\delta \widetilde{\Psi}\rangle=Q_B|\widetilde{\Lambda}\rangle+\sum_{n=0}^{\infty} \frac{1}{n !}\left[\Psi^n \Lambda\right],   \end{equation}
where $\left[\Psi^n \Lambda\right]:=[\Psi\Psi\dots\Psi\Lambda]\in\widetilde{\s{H}}_T$ is defined by the relation 
\begin{equation}\label{eq:[]stdef}
\left\langle\Psi\left|c_0^{-} \right| [\Psi_1\Psi_2\dots\Psi_n]\right\rangle=\{\Psi\Psi_1\Psi_2\dots\Psi_n\},\quad \forall~~\Psi\in\widehat{\s{H}}_T.
\end{equation}
We now list some of the properties of the $n$-string vertex $\{\Psi_1\Psi_2\dots\Psi_n\}$ and the state $[\Psi_1\Psi_2\dots\Psi_n]$.
\begin{enumerate}
    \item 
Ghost number anomaly implies that the a correlation function of vertex operators on a genus $g$ Riemann surface is nonzero only if the total ghost number of vertex operators in the correlator is $6-6g$. This implies that the $n$-string vertex $\{\Psi_1\Psi_2\dots\Psi_n\}$ vanishes unless \cite[Eq. (3.7)]{deLacroix:2017lif} 
\begin{equation}
    \sum_{i=1}^n(\text{gh}(\Psi_i)-2)=0~,
\end{equation}
where $\text{gh}(\Psi_i)$ denotes the ghost number of $\Psi_i$. This means that on a genus $g$ surface
\begin{equation}\label{eq:ghostn[]}
    \text{gh}\left([\Psi_1\Psi_2\dots\Psi_n]\right)=3-6g+\sum_{i=1}^n(\text{gh}(\Psi_i)-2).
\end{equation}
\item \label{prop:grpr[]}
The grassmann parity of $[\Psi_1\Psi_2\dots\Psi_n]$ is opposite to the sum of the grassmann parities of $\Psi_i$'s. 
\item 
The following relations hold \cite{Sen:2015uaa}: 
\begin{equation}\label{eq:commin{}}
\begin{gathered}
\left\{\Psi_1 \Psi_2 \cdots \Psi_{i-1} \Psi_{i+1} \Psi_i \Psi_{i+2} \cdots \Psi_n\right\}=(-1)^{\gamma_i \gamma_{i+1}}\left\{\Psi_1 \Psi_2 \cdots \Psi_n\right\}~, \\
{\left[\Psi_1 \cdots \Psi_{i-1} \Psi_{i+1} \Psi_i \Psi_{i+2} \cdots \Psi_n\right]=(-1)^{\gamma_i \gamma_{i+1}}\left[\Psi_1 \cdots \Psi_n\right]~,}
\end{gathered}
\end{equation}
where $\gamma_i$ is the grassmannality of $\left|\Psi_i\right\rangle$. 
\item 
The following identity, called the \textit{main identity}, is satisfied \cite{Sen:2015uaa}:
\begin{equation}\label{eq:mainid{}}
\begin{aligned}
& \sum_{i=1}^n(-1)^{\gamma_1+\dots+ \gamma_{i-1}}\left\{\Psi_1 \cdots \Psi_{i-1}\left(Q_B \Psi_i\right) \Psi_{i+1} \cdots \Psi_n\right\} \\
= & -\frac{1}{2} \sum_{\substack{\ell, k \geq 2 \\
\ell+k=n}} \sum_{\substack{\left\{i_a ; a=1, \dots,\ell\right\},\left\{j_b ; b=1, \dots, k\right\} \\
\{i a\} \cup\left\{j_b\right\}=\{1, \dots, n\}}} \sigma\left(\left\{i_a\right\},\left\{j_b\right\}\right)\left\{\mathcal{G}\left[\Psi_{j_1} \cdots \Psi_{j_k}\right]\Psi_{i_1} \cdots \Psi_{i_{\ell}}\right\}~,
\end{aligned}
\end{equation}
and
\begin{equation}\label{eq:mainid[]}
\begin{aligned}
& Q_B\left[\Psi_1 \cdots \Psi_n\right]+\sum_{i=1}^n(-1)^{\gamma_1+\dots +\gamma_{i-1}}\left[\Psi_1 \cdots \Psi_{i-1}\left(Q_B \Psi_i\right) \Psi_{i+1} \cdots \Psi_n\right] \\
= & -\sum_{\substack{\ell, k \geq 2 \\
\ell+k=n}} \sum_{\substack{\left\{i_a ; a=1, \dots, \ell\right\},\left\{j_b ; b=1, \dots, k\right\} \\
\left\{i_a\right\} \cup\left\{j_b\right\}=\{1, \dots, n\}}} \sigma\left(\left\{i_a\right\},\left\{j_b\right\}\right)\left[\mathcal{G}\left[\Psi_{j_1} \cdots \Psi_{j_k}\right]\Psi_{i_1} \cdots \Psi_{i_{\ell}} \right]~,
\end{aligned}
\end{equation}
where $\sigma\left(\left\{i_a\right\},\left\{j_b\right\}\right)$ is the sign that we get while rearranging 
\begin{equation}
\Psi_1, \cdots \Psi_n\to\Psi_{j_1}, \cdots \Psi_{j_k}\Psi_{i_1}, \cdots \Psi_{i_{\ell}}, ~.   \end{equation}
Note in the inner sum above, the sum is only over partitions without permutations over the individual subpartitions. For example for $n=4$, the sum is over 
\begin{equation}
    \begin{split}
 \{\{1,2\},\{3,4\}\},\{\{1,3\},\{2,4\}\},\{\{1,2\},\{3,4\}\}, \{\{1\},\{2,3,4\}\},\{\{2\},\\\{1,3,4\}\},\{\{3\},\{1,2,4\}\},\{\{4\},\{1,2,3\}\},       
    \end{split}
\end{equation}
see the notion of unshuffles in \eqref{eq:unshuff} for a precise definition.
\item The following relation holds \cite{Sen:2015uaa}:
\begin{equation}\label{eq:[]in{}}
\begin{split}
\left\{\Psi_1 \cdots \Psi_k \mathcal{G}\left[\Psi'_1 \cdots \Psi'_{\ell}\right]\right\}&=(-1)^{\gamma+\gamma'+\gamma \gamma'}\left\{\Psi'_1 \cdots \Psi'_{\ell} \mathcal{G}\left[\Psi_1 \cdots \Psi_k\right]\right\}~,
\end{split}
\end{equation}
where $\gamma$ and $\gamma'$ are the total grassmannalities of $\Psi_1, \cdots \Psi_k$ and $\Psi'_1, \cdots \Psi'_{\ell}$ respectively. 
\end{enumerate}
The equation of motion is given by 
\begin{equation}
\begin{aligned}
& Q_B(|\Psi\rangle-\mathcal{G}|\widetilde{\Psi}\rangle)=0~, \\
& Q_B|\widetilde{\Psi}\rangle+\sum_{n=1}^{\infty} \frac{1}{(n-1) !}\left[\Psi^{n-1}\right]=0~.
\end{aligned}
\end{equation}
Acting on the second equation by $\mathcal{G}$ from the left and adding it to the first equation we get
\begin{equation}
Q_B|\Psi\rangle+\sum_{n=1}^{\infty} \frac{1}{(n-1) !} \mathcal{G}\left[\Psi^{n-1}\right]=0~.
\end{equation}
For classical action, the sum starts from $n=3$ since the one and two point function on sphere vanishes for all SCFTs. So the classical action, gauge transformation and equations of motion takes the form
\begin{equation}\label{eq:classaceomgtrans}
\begin{split}
\text{Action:}\quad&S_{\text{cl}}=\frac{1}{g_s^2}\left[-\frac{1}{2}\left\langle\widetilde{\Psi}\left|c_0^{-} Q_B \s{G}\right| \widetilde{\Psi}\right\rangle+\left\langle\widetilde{\Psi}\left|c_0^{-} Q_B\right| \Psi\right\rangle+\sum_{n=3}^{\infty} \frac{1}{n !}\left\{\Psi^n\right\}_0\right]~, \\\text{Gauge transformation:}\quad& |\delta \Psi\rangle=Q_B|\Lambda\rangle+\sum_{n=2}^{\infty} \frac{1}{n !} \mathcal{G}\left[\Psi^n \Lambda\right]_0, \quad|\delta \widetilde{\Psi}\rangle=Q_B|\widetilde{\Lambda}\rangle+\sum_{n=2}^{\infty} \frac{1}{n !}\left[\Psi^n \Lambda\right]_0,\\\text{Eq. of motion:}\quad & Q_B(|\Psi\rangle-\mathcal{G}|\widetilde{\Psi}\rangle)=0,\quad  Q_B|\widetilde{\Psi}\rangle+\sum_{n=3}^{\infty} \frac{1}{(n-1) !}\left[\Psi^{n-1}\right]_0=0, \\\implies \quad&Q_B|\Psi\rangle+\sum_{n=3}^{\infty} \frac{1}{(n-1) !} \mathcal{G}\left[\Psi^{n-1}\right]_0=0, 
\end{split}    
\end{equation}
where $\left\{\Psi^n\right\}_0$ denotes the genus zero $n$-string vertex and $[\Psi^n]_0$ is defined as in \eqref{eq:[]stdef} with the genus zero string vertex. It also satisfies the identities in \eqref{eq:mainid[]} and \eqref{eq:mainid{}}. We will omit the subscript 0 from now on. 
\subsection{$g$-twisted $L_\infty$-algebras}\label{sec:genLinfalg}
In this section, we slightly generalise the mathematical definition of $L_\infty$-algebras suitable for application to closed superstring field theory. We refer the reader to \cite{10.1007/BFb0101184,Zwiebach:1992ie} for material on $L_\infty$-algebras. \par For $p, q \in \mathbb{N}$, a $(p, q)$-unshuffle is a permutation
\begin{equation}\label{eq:unshuff}
\left(\mu_1, \cdots, \mu_p, \nu_1, \cdots, \nu_q\right)~,
\end{equation}
of $(1,2, \dots, p+q)$ such that
\begin{equation}
\mu_1<\mu_2<\cdots<\mu_p,\quad \nu_1<\nu_2<\cdots<v_q. 
\end{equation}
We denote the set of all $(p,q)$-unshuffles by UnShuff$(p,q)$. \\
Let $V=\oplus_{i\in\Z}V_i$ be a $\mathbb{Z}$-graded vector space. For a homogenous vector $v$, we denote the degree by $|v|$.
Let 
\begin{equation}
\chi\left(\sigma, v_1, \cdots, v_n\right) \in\{-1,+1\}~,
\end{equation}
denote the products of $(-1)^{\left|v_i\right|\left|v_{i+1}\right|}$ for each interchange of neighbours 
\begin{equation}
\left(\dots, v_i, v_{i+1}, \dots\right)\longrightarrow\left(\dots, v_{i+1}, v_i, \dots\right)~,    
\end{equation}
involved in a decomposition of the permutation $\sigma$ into transpositions. Note that this is well defined since the number of transposition in a decomposition of a permutation is either even or odd.  \\Let $V^{\wedge n}$ denote the symmetrized tensor product obtained by linear span of vectors of the form 
\begin{equation}\label{eq:symm_tensor_vn}
    v_1\wedge\dots\wedge v_n:=\sum_{\sigma\in S_n}\chi(\sigma,v_1,\dots,v_n)(v_{\sigma(1)}\otimes\dots\otimes v_{\sigma(n)}).
\end{equation}
Clearly for any $\sigma\in S_n$
\begin{equation}\label{eq:gradsignprop}
v_{\sigma(1)}\wedge\dots\wedge v_{\sigma(n)}=\chi(\sigma,v_1,\dots,v_n)( v_1\wedge\dots\wedge v_n).    
\end{equation}
The symmetrized tensor product $V^{\wedge n}$ naturally acquires a grading from that of $V$ by the condition
\begin{equation}
    V_{i_1}\wedge\dots\wedge V_{i_n}\subset (V^{\wedge n})_{i_1+\dots+i_n}.
\end{equation}
A multilinear map $f:V^{\times n}\longrightarrow V$ is   
called \textit{graded-symmetric} if
\begin{equation}
f(v_{\sigma(1)},\dots, v_{\sigma(n)})=\chi(\sigma,v_1,\dots,v_n)( v_1,\dots, v_n),\quad \sigma\in S_n~.    
\end{equation}
By \eqref{eq:gradsignprop}, any 
multilinear map $f:V^{\wedge n}\longrightarrow V$ is automatically graded-symmetric. 
Let $c_j:V^{\wedge j}\longrightarrow V$ and $d_i:V^{\wedge i}\longrightarrow V$ be multilinear symmetric maps. We define the product $c_jd_i:V^{\wedge (i+j-1)}\longrightarrow V$ as follows: for $(i+j-1=n)$-tuples $\left(v_1, \cdots, v_n\right)$ of homogeneous vectors $v_i \in V_{\left|v_i\right|}$ 
\begin{equation}\label{eq:prodcidi}
\begin{split}
c_jd_i(v_1, \cdots, v_n)=\sum_{\sigma \in \operatorname{UnShuff}(i, j-1)} &\chi\left(\sigma, v_1, \cdots, v_n\right) \\& \times   
c_j\left(d_i\left(v_{\sigma(1)}, \cdots, v_{\sigma(i)}\right)), v_{\sigma(i+1)}, \cdots, v_{\sigma(n)}\right).
\end{split}
\end{equation}
This product can equivalently be written as 
\begin{equation}
    c_j d_i=c_j\circ\left(d_i \wedge \mathds{1}_{V^{\wedge j-1}}\right),
\end{equation}
where 
for multilinear maps $\alpha: V^{\wedge k} \longrightarrow V^{\wedge l},~ \beta: V^{\wedge m} \longrightarrow V^{\wedge n}$, we can define their wedge product
\begin{equation}
\alpha \wedge \beta: V^{\wedge k+m} \longrightarrow V^{\wedge l+n}~,
\end{equation}
by writing
\begin{equation}\label{eq:albewed}
\begin{aligned}
\alpha \wedge \beta\left(v_1, \ldots, v_{k+m}\right)=\sum_{\sigma \in \operatorname{UnShuff}(k, m)} &\chi\left(\sigma, v_1, \cdots, v_{k+m}\right)\times \\&    
\alpha\left(v_{\sigma(1)}, \cdots, v_{\sigma(k)}\right)\wedge\beta\left( v_{\sigma(k+1)}, \cdots, v_{\sigma(k+m)}\right), 
\end{aligned}
\end{equation}
and 
\begin{equation}
    \mathds{1}_{V^{\wedge n}}=\frac{1}{n !}\left(\mathds{1}_{V}\right)^{\wedge n}=\left(\mathds{1}_V\right)^{\otimes n}~,
\end{equation}
is the identity operator on $V^{\wedge n}$. We also define the commutator of two multilinear maps by 
\begin{equation}\label{eq:gradcommdef}
    [c_j,d_i]:=c_jd_i-(-1)^{|c_j||d_i|}d_ic_j~,
\end{equation}
where $|c_j|$ is the degree of \footnote{Recall that a map $f:V\to W$ between $\Z$-graded vector spaces is of degree $n$ if $f(V_i)\subset W_{i+n}$ for all $i\in\Z$, where $V_i$ is the degree $i$ homogenous subspace.}  the map $c_j$.\\\\
Let $V^{\star}$ be the dual of $V$ with some choice of $\Z$-grading. 
We now define twisted $L_\infty$-algebras which are slight generalisation of classical $L_\infty$-algebras.
\begin{defn}\label{def:g_twisted_L_infty}
Let $V$ be a $\Z$-graded vector space with dual $V^\star$ and let $g:V^\star\longrightarrow V$ be a linear map. A $g$-twisted $L_{\infty}$-algebra is a tuple $(V,\{\ell_n\}_{n=1}^\infty,g)$ where $\ell_1:V\oplus V^\star\longrightarrow V\oplus V^\star$, we will denote the component maps by the same symbol $\ell_1:V\longrightarrow V,~\ell_1:V^\star\longrightarrow V^\star$,  and  
$\{\ell_n\}_{n=2}^\infty$
is a family of multilinear maps called the $n$-ary bracket: 
\begin{equation}
\ell_n(\cdots):=[-,-, \dots,-]_n:V^{\wedge n}\longrightarrow V^{\star}~.
\end{equation}
The following conditions hold:
\begin{enumerate}
\item  $\ell_1$ has degree $-1$ and is nilpotent: $\ell_1^2=0$.
\item The $n$-ary brackets $\ell_n$ are of degree $n-2$.
    \item 
\textit{Graded-symmetry}: each $\ell_n$ is graded symmetric: for homogeneous vectors $v_i$ and a permutation $\sigma$
\begin{equation}\label{eq:grantsymln}
\ell_n\left(v_{\sigma(1)}, v_{\sigma(2)}, \cdots, v_{\sigma(n)}\right)=\chi\left(\sigma, v_1, \cdots, v_n\right) \cdot \ell_n\left(v_1, v_2, \cdots v_n\right).
\end{equation}
\item \textit{Strong homotopy Jacobi identity}: for all $k \in \mathbb{N}$, 
\begin{equation}
\sum_{\substack{i,j\in\N\\i+j=k+1}}\ell_j^g \ell^g_i=0,    
\end{equation}
where $\ell_i^g:V^{\wedge i}\longrightarrow V$ is given by
\begin{equation}
    \ell^g_1:=\ell_1,~\ell^g_i:=g\circ\ell_i,
\end{equation}
and the product is defined in \eqref{eq:prodcidi}.
\end{enumerate}
\end{defn}
\noindent The first few strong homotopy Jacobi identities are 
\begin{equation}\label{eq:stjacidexp}
\begin{split}
k=1:~~&\ell_1\left(\ell_1\left(v_1\right)\right)=0, \\
k=2:~~&\ell_1\left(\ell_2^g\left(v_1, v_2\right)\right)+\ell_2^g\left(\ell_1\left(v_1\right), v_2\right)+(-1)^{|v_1| |v_2|} \ell_2^g\left(\ell_1\left(v_2\right), v_1\right)=0, \\
k=3:~~&\ell_1\left(\ell_3^g\left(v_1, v_2, v_3\right)\right)+\ell_2^g\left(\ell_2^g\left(v_1, v_2\right), v_3\right)\\&+ 
(-1)^{|v_1|\left(|v_2|+|v_3|\right)} \ell_2^g\left(\ell_2^g\left(v_2, v_3\right), v_1\right)\\&+ (-1)^{|v_3|\left(|v_1|+|v_2|\right)} \ell_2^g\left(\ell_2^g\left(v_3, v_1\right), v_2\right)\\&+\ell_3^g\left(\ell_1\left(v_1\right), v_2, v_3\right)+(-1)^{|v_1|} \ell_3^g\left(v_1, \ell_1\left(v_2\right), v_3\right)\\
&+ (-1)^{|v_1|+|v_2|} \ell_3^g\left(v_1, v_2, \ell_1\left(v_3\right)\right)=0.
\end{split}
\end{equation}
It is easy to see that if $(V,\{\ell_n\},g)$ is a $g$-twisted $L_\infty$-algebra then $(V,\{\ell^g_n\})$ is a usual $L_\infty$-algebra. One can generalise the notion of $g$-twisted $L_\infty$-algebra to the case when $V$ is a supervector space, i.e. a $\mathbb{Z}_2$-graded vector space. Let us denote the $\mathbb{Z}_2$-degree of a vector $v$ by $d_{\mathbb{Z}_2}(v)\in\{\bar{0},\bar{1}\}=\Z_2$.
\begin{defn}
A $g$-twisted $L_\infty$-superalgebra is a triple $(V,\{\ell_n\},g)$, where $V$ is a supervector space and $\ell_n,g$ are as above, satisfying the axioms of $g$-twisted $L_\infty$-algebra with the sign $\chi(\sigma,v_1,\dots,v_n)\in\{+1,-1\}$ obtained by replacing the degree $|v|$ of a vector $v$ by\footnote{Note that in the sum in \eqref{eq:moddegsupvec}, $d_{\mathbb{Z}_2}(v)$ is considered as the integer $0,1$ rather than an equivalence class.} 
\begin{equation}\label{eq:moddegsupvec}
    d(v):=|v|+(d_{\mathbb{Z}_2}(v)+1).
\end{equation}
The degree of $\ell_n$ is $(n-2)$ with respect to the $\Z$-grading $|\cdot|$ on $V$ and is odd degree with respect to $d(\cdot)$.
\end{defn}
\begin{remark}\label{rem:lilj=1/2[]}
Suppose each $\ell_n^g$ has odd degree $d(\ell_n^g)$ as a map from $V\to V$, then the strong homotopy Jacobi identity can also be written 
as
\begin{equation}\label{eq:lilj=1/2[]}
\sum_{\substack{i,j\in\N\\i+j=k+1}}\ell_j^g \ell^g_i=\frac{1}{2}\sum_{\substack{i,j\in\N\\i+j=k+1}}[\ell_j^g ,\ell^g_i]=0.    
\end{equation}
\end{remark}
From now on, we assume that $V$ is a supervector space. To reduce to the usual case, we can take the odd subspace of $V$ to be trivial.
There is a natural pairing 
\begin{equation}
    \begin{split}
        &V^\star\otimes V\longrightarrow \C\\&(\phi,v)\mapsto \phi(v).
    \end{split}
\end{equation}
Suppose now that $V$ is a reflexive vector space\footnote{Recall that a vector space is called reflexive if the natural map $V\to V^{\star\star}$ given by $x\mapsto[\widehat{x}:f\mapsto \widehat{x}(f):=f(x)]$ is a linear isomorphism. Thus finite dimensional vector spaces are reflexive.\label{foot:refl}}. Define a bilinear pairing $\omega:V^\star\otimes V \longrightarrow \C$ by
\begin{equation}\label{eq:omdef}
    \omega(v_1,v_2):=(-1)^{d(v_1)}v_1(v_2)~,
\end{equation}
on homogeneous elements $v_1\in V^\star$ and any $v_2\in V$ and then extend it linearly to all of $V^\star$.
The bilinear pairing $\omega:V^\star\otimes V \longrightarrow \C$ defined above is said to be graded anti-symmetric if  
\begin{equation}
\omega\left(v_1, v_2\right)=-(-1)^{d(v_1)d(v_2)} \omega\left(v_2, v_1\right),
\end{equation}
where the right hand side is the bilinear pairing on $V^{\star\star}\otimes V^\star$ with the natural identification of $V^{\star\star}$ with $V$, see footnote \ref{foot:refl}. 
The $n$-ary brackets $\ell_n,~n\geq 2$ are called \textit{cyclic} with respect to $\omega$ if they satisfy 
\begin{equation}
\omega\left(\ell_n\left(v_1, \dots, v_n\right), v_{n+1}\right) =-(-1)^{d(v_1)}\omega\left(v_1, \ell_n\left(v_2, \ldots, v_{n+1}\right)\right) .    
\end{equation}
We say that $(V,\{\ell_n\}_{n=1}^\infty,g)$ is a cyclic $g$-twisted $L_\infty$-(super)algebra with respect to $\omega$ if every $\ell_n,~n\geq 2$ is cyclic with respect to $\omega$. 
\subsection{Closed superstring field theory and $L_\infty$-superalgebra}\label{sec:2.4}
We now show that the superstring theory Hilbert space $\widehat{\s{H}}_T$ along with the $n$-string vertices is a cyclic $L_\infty$-superalgebra with respect to (modified) BPZ inner product. Let us first define the grading on the full Hilbert space $\s{H}_T$ by declaring 
\begin{equation}
    |\Psi|=1-\text{gh}(\Psi)~,
\end{equation}
where $\text{gh}(\Psi)$ is the ghost number of $\Psi$. 
Next the supervector space structure of $\widehat{\s{H}}_T$ is defined by putting a $\mathbb{Z}_2$ grading on the full Hilbert space $\s{H}_T$ as follows:
\begin{equation}
    d_{\mathbb{Z}_2}(\Psi)=\begin{cases}
        0,&\Psi\in\s{H}_{\mathrm{NSNS}}, \s{H}_{\mathrm{RR}}\quad \text{for Type II},\\0,&\Psi\in \s{H}_{\mathrm{NS}}\quad \text{for Heterotic},\\1, &\Psi\in\s{H}_{\mathrm{NSR}}, \s{H}_{\mathrm{RNS}}\quad  \text{for Type II},\\ 1, &\Psi\in\s{H}_{\mathrm{R}}\quad \text{for Heterotic}.
    \end{cases}
\end{equation}
With this definition, it is clear that the grassmannality of states in $\widehat{\s{H}}_T$ and $\widetilde{\s{H}}_T$ and the grading is related by 
\begin{equation}\label{eq:dpsidefdz2}
    d(\Psi):=|\Psi|+d_{\mathbb{Z}_2}(\Psi)+1=d_{\mathbb{Z}_2}(\Psi)-\text{gh}(\Psi)+2\equiv\gamma \bmod 2~,
\end{equation}
where $\gamma$ is the grassmanality of $\Psi$. 
Next, the relations \eqref{eq:inprodphirc} identifies $\widetilde{\s{H}}_T$ with a subspace of the dual space\footnote{Note that for finite dimensional vector spaces, this identification is an isomorphism. But since our spaces are necessarily infinite dimensional this defines only a subspace.} $\widehat{\s{H}}_T^\star$ via the map 
\begin{equation}\label{eq:HtilidHhat*}
    \widetilde{\s{H}}_T\ni\widetilde{\Psi}\mapsto \left\langle \widetilde{\Psi}|c_0^-|\cdot\right\rangle,
\end{equation}
where $\left\langle \widetilde{\Psi}|c_0^-|\cdot\right\rangle$ is a linear functional on $\widehat{\s{H}}_T$ acting via the BPZ inner product: 
\begin{equation}
\widehat{\s{H}}_T\ni\Psi\mapsto \left\langle \widetilde{\Psi}|c_0^-|\Psi\right\rangle\in\C.    
\end{equation}
Moreover since Hilbert spaces are reflexive i.e. $\widehat{\s{H}}_T\cong \widehat{\s{H}}_T^{\star\star}$ under the map (see footnote \ref{foot:refl}) 
\begin{equation}\label{eq:Hwh**iso}
\begin{split}
    &\Psi\mapsto \langle\cdot|c_0^-|\Psi\rangle\in \widehat{\s{H}}_T^{\star\star}~,
\end{split}
\end{equation}
which acts as 
\begin{equation}
\wt{\Psi}\mapsto \left\langle\wt{\Psi}|c_0^-|\Psi\right\rangle~,    
\end{equation}
the second relation in \eqref{eq:inprodphirc} can be understood as the restriction of the pairing  $\widehat{\s{H}}_T^{\star\star}\otimes \widehat{\s{H}}_T^{\star}\to\mathbb{C}$.
\par 
Now
define the multilinear maps $\ell_1:=Q_B$ and (see \eqref{eq:[]stdef})
\begin{equation}
\begin{split}
    &\ell_n:\widehat{\s{H}}_T^{\wedge n}\longrightarrow\widetilde{\s{H}}_T~,\\&\Psi_1\wedge\dots\wedge\Psi_n\mapsto [\Psi_1\Psi_2\dots\Psi_n],\quad n>1.
\end{split}
\end{equation}
Since $Q_B$ has ghost number $-1$, we have
\begin{equation}
    |Q_B\Psi|=1-\text{gh}(Q_B\Psi)=1-(\text{gh}(\Psi)+1)=(1-\text{gh}(\Psi))-1=|\Psi|-1.
\end{equation}
Thus $\ell_1$ is a degree $-1$ map.
Next, from \eqref{eq:ghostn[]} we see that 
\begin{equation}\label{eq:|ln|calc}
\begin{split}
 |\ell_n(\Psi_1,\dots,\Psi_n)|&=1-\text{gh}(\ell_n(\Psi_1,\dots,\Psi_n))\\&=1-3-\sum_{i=1}^n(\text{gh}(\Psi_i)-2)\\&=(n-2)+\sum_{i=1}^n(1-\text{gh}(\Psi_i))\\&=(n-2)+\sum_{i=1}^n|\Psi_i|,   
\end{split}
\end{equation}
which means that $\ell_n$ is a degree $n-2$ map as required. We also define the $\Z_2$-parity of the string product to be\footnote{This is consistent with the fact that the product of two NSR or RNS states in Type II is an NSNS or RR state. While the product of an NSNS or RR state with an NSR or RNS state is NSR or RNS. Similarly the product of two R states is an NS state in Heterotic while the product of an NS state with an R state is an R state.} 
\begin{equation}\label{eq:dz2lndef}
    d_{\Z_2}(\ell_n(\Psi_1,\dots,\Psi_n)):=\sum_{i=1}^nd_{\Z_2}(\Psi_i)\in\mathbb{Z}_2~,
\end{equation}
so that 
\begin{equation}
    d(\ell_n(\Psi_1,\dots,\Psi_n))=|\ell_n(\Psi_1,\dots,\Psi_n)|+\sum_{i=1}^nd_{\Z_2}(\Psi_i)+1.
\end{equation}
This implies that 
\begin{equation}\label{eq:dlncalc}
  d(\ell_n(\Psi_1,\dots,\Psi_n))=\sum_{i=1}^nd(\Psi_i)-1~,
\end{equation}
so that the degree of $\ell_n$ with respect to $d(\cdot)$ is $-1$ and $\ell_n$ is an odd degree operator.\\\\
From \eqref{eq:dpsidefdz2} and \eqref{eq:|ln|calc}, we see that the grassmanality of $\ell_n(\Psi_1,\dots,\Psi_n)$ is opposite of the sum of the grassmanalities of $\Psi_i$:
\begin{equation}
    \gamma_{\ell_n}\equiv1+\sum_{i=1}^n\gamma_i\bmod 2
\end{equation}
where $\gamma_{\ell_n}$ is the grassmanality of $\ell_n(\Psi_1,\dots,\Psi_n)$.
From \eqref{eq:commin{}},
we see that $\ell_n$ is graded-symmetric multilinear map. Indeed, decomposing any permutation $\sigma$ into transpositions and then using \eqref{eq:commin{}} gives the required relation. Finally, an application of the PCO $\mathcal{G}$ on the main identity \eqref{eq:mainid[]} gives
\begin{equation}
\begin{aligned}
& Q_B\mathcal{G}\left[\Psi_1 \cdots \Psi_n\right]+\sum_{i=1}^n(-1)^{\gamma_1+\dots +\gamma_{i-1}}\mathcal{G}\left[\Psi_1 \cdots \Psi_{i-1}\left(Q_B \Psi_i\right) \Psi_{i+1} \cdots \Psi_n\right] \\
& +\sum_{\substack{\ell, k \geq 2 \\
\ell+k=n}} \sum_{\substack{\left\{i_a ; a=1, \dots, \ell\right\},\left\{j_b ; b=1, \dots, k\right\} \\
\left\{i_a\right\} \cup\left\{j_b\right\}=\{1, \dots, n\}}} \sigma\left(\left\{i_a\right\},\left\{j_b\right\}\right)\mathcal{G}\left[\mathcal{G}\left[\Psi_{j_1} \cdots \Psi_{j_k}\right]\Psi_{i_1} \cdots \Psi_{i_{\ell}} \right]=0,
\end{aligned}
\end{equation}
where we used the fact that $[Q_B,\s{G}]=0$. This  
is exactly the same as the strong homotopy Jacobi identity once we realise that the sum over partition of $\{1,\dots,n\}$ in \eqref{eq:mainid[]} is the same as the sum over unshuffles. The first two terms in the above identity correspond to the terms $(i,j)=(1,n), (n,1)$.  
The bilinear pairing between $\widetilde{\s{H}}_T$ and $\widehat{\s{H}}_T$ is given by 
\begin{equation}
    \omega\left(\widetilde{\Psi},\Psi\right)=(-1)^{d(\widetilde{\Psi})}\left\langle\widetilde{\Psi}\big{|}c_0^-\big{|}\Psi\right\rangle. 
\end{equation}
We also have 
\begin{equation}
\omega\left(\Psi,\widetilde{\Psi}\right)=(-1)^{d(\Psi)}\left\langle\Psi\big{|}c_0^-\big{|}\widetilde{\Psi}\right\rangle.     
\end{equation}
The BPZ inner product is symmetric up to a sign dependent on grassmanalities of the states:
\begin{equation}
    \langle\widetilde{\Psi}|c_0^-|\Psi'\rangle=-(-1)^{\wt{\gamma}+\gamma'+\wt{\gamma}\gamma'}\langle\Psi'|c_0^-|\widetilde{\Psi}\rangle~,
\end{equation}
where $\wt{\gamma}$ and $\gamma'$ are the grassmannality of $\wt{\Psi}\in\wt{\s{H}}_T$ and $\Psi'\in \wh{\s{H}}_T$ respectively. The factor comes from commuting the vertex operators of the states past $c_0^-$ and past each other and the additional negative sign comes from the SL(2,$\mathbb{C}$)-action $I:z\to -1/z$ on the integral for the zero mode of $c(z),\overline{c}(\bar{z})$.  
We have from \eqref{eq:dpsidefdz2}
\begin{equation}
\omega\left(\widetilde{\Psi},\Psi\right)=-(-1)^{2d(\widetilde{\Psi})+d(\Psi)} (-1)^{d(\widetilde{\Psi})d(\Psi)}\omega\left(\Psi,\widetilde{\Psi}\right).   
\end{equation}
This implies that $\omega$ is graded anti-symmetric:
\begin{equation}\label{eq:omgradantcomm}
\omega\left(\widetilde{\Psi},\Psi\right)=-(-1)^{d(\widetilde{\Psi})d(\Psi)}\omega\left(\Psi,\widetilde{\Psi}\right).    
\end{equation}
From \eqref{eq:Ghermbpz} we see that 
\begin{equation}\label{eq:Ghermom}
\omega\left(\widetilde{\Psi},\s{G}\wt{\Psi}'\right)=\omega\left(\s{G}\wt{\Psi},\widetilde{\Psi}'\right)~.    
\end{equation}
We now show that with respect to $\omega$, $\ell_n$ is cyclic:
\begin{equation}\label{eq:cycln}
\omega\left(\ell_n\left(\Psi_1 ,\Psi_2,\dots, \Psi_n\right), \Psi_{n+1}\right)=-(-1)^{d(\Psi_1)} \omega\left(\Psi_1, \ell_n\left(\Psi_2, \dots,  \Psi_{n+1}\right)\right)~,
\end{equation}
for $\Psi_i \in \widehat{\mathcal{H}}_T$. 
This is proved as follows:
\begin{equation}
\begin{aligned}
\omega\left(\Psi_1, \ell_n \left(\Psi_2,\dots, \Psi_{n+1}\right)\right)&=(-1)^{d(\Psi_1)}\left\langle \Psi_1|c_0^-| \ell_n \left(\Psi_2 \ldots \Psi_{n+1}\right)\right\rangle\\&=(-1)^{d(\Psi_1)}\left\{\Psi_1 \Psi_2 \ldots \Psi_{n+1}\right\} \\
& =(-1)^{d(\Psi_1)}(-1)^{\gamma_1 \gamma_{n+1}+\ldots+\gamma_n \gamma_{n+1}}\left\{\Psi_{n+1} \Psi_1\Psi_2 \Psi_2 \ldots \Psi_n\right\} \\
& =(-1)^{d(\Psi_1)}(-1)^{(\gamma_1 +\dots+ \gamma_n) \gamma_{n+1}}(-1)^{d(\Psi_{n+1})} \\&\times\omega\left(\Psi_{n+1}, \ell_n \left(\Psi_1, \dots, \Psi_n\right)\right), 
\end{aligned}
\end{equation}
where we used \eqref{eq:commin{}}. Here $(-1)^{\gamma_i}$ is the grassmannality of $\Psi_i$. Now using \eqref{eq:dpsidefdz2} and \eqref{eq:dlncalc} it is easy to see that
\begin{equation}
    \gamma_1+\dots+\gamma_n\equiv (d(\ell_n(\Psi_1,\dots,\Psi_n))-1) \bmod 2
\end{equation}
This gives using \eqref{eq:dpsidefdz2} 
\begin{equation}
\begin{split}
    \omega\left(\Psi_1, \ell_n \left(\Psi_2,\dots, \Psi_{n+1}\right)\right)&= (-1)^{d(\Psi_1)} (-1)^{(d(\ell_n(\Psi_1,\dots,\Psi_n)-1)d(\Psi_{n+1})}(-1)^{d(\Psi_{n+1})}\\&\times\omega\left( \Psi_{n+1},\ell_n \left(\Psi_1, \dots,\Psi_n\right)\right)\\&=-(-1)^{d(\Psi_1)}\omega\left(\ell_n \left(\Psi_1, \dots,\Psi_n\right),\Psi_{n+1}\right)
\end{split}    
\end{equation}
where we used \eqref{eq:omgradantcomm}.
\par The superstring action can then be written using $\omega$ as 
\begin{equation}\label{eq:clactellop}
    S_{\text{cl}}=\frac{1}{g_s^2}\left[-\frac{1}{2}\omega(\widetilde{\Psi},Q_B\s{G}\widetilde{\Psi})+\omega(\widetilde{\Psi},Q_B\Psi)+\sum_{n=1}^\infty\frac{1}{(n+1)!}\omega\left(\Psi,\ell_n(\Psi^n)\right)\right],
\end{equation}
where $\widetilde{\Psi}\in\widetilde{\s{H}}_T,\Psi\in\widehat{\s{H}}_T$ are grassmann even string fields and 
\begin{equation}
 \ell_n(\Psi^n):=\ell_n(\underbrace{\Psi,\Psi,\dots,\Psi}_{n~\text{times}}).   
\end{equation}
Note that the first term in the sum is zero since the BPZ inner product vanishes for two states in $\wh{\s{H}}_T$. 
Choosing a smooth interpolation $\Psi(t)$ for $0\leq t\leq 1$ with $\Psi(0)=0$ and $\Psi(1)=\Psi$, we can write the action as 
\begin{equation}\label{eq:stactionint}
S_{\text{cl}}=\frac{1}{g_s^2}\left[-\frac{1}{2}\omega(\widetilde{\Psi},Q_B\s{G}\widetilde{\Psi})+\omega(\widetilde{\Psi},Q_B\Psi)+\sum_{n=1}^\infty\frac{1}{n!} \int_0^1 dt~\omega(\dot{\Psi}(t),\ell_n(\Psi(t)^n))\right]~.   
\end{equation}
The EOM takes the form
\begin{equation}
\begin{aligned}
& \operatorname{EOM}(\widetilde{\Psi})=Q_B(\Psi-\s{G} \widetilde{\Psi}), \\
& \operatorname{EOM}(\Psi)=Q_B \widetilde{\Psi}+\sum_{n=2}^{\infty} \frac{1}{n !} \ell_n \left(\Psi^{ n}\right)=Q_B \widetilde{\Psi}+\mathcal{J}^{\prime}(\Psi),
\end{aligned}
\end{equation}
where we have introduced the notation
\begin{equation}
\mathcal{J}^{\prime}(\Psi)=\sum_{n=2}^{\infty} \frac{1}{n !} \ell_n \left(\Psi^{n}\right) .
\end{equation}
Using these two equations, we obtain
\begin{equation}\label{eq:EOMPsi}
\operatorname{EOM}(\Psi)=Q_B \Psi+\mathcal{J}(\Psi)~,
\end{equation}
where 
\begin{equation}
\mathcal{J}(\Psi):=\s{G} \s{J}^{\prime}(\Psi)=\sum_{n=2}^{\infty} \frac{1}{n !} \ell_n^{\s{G}} \left(\Psi^{n}\right),    
\end{equation}
where $\ell_n^{\s{G}}=\s{G}\circ\ell_n$.
The linearized gauge transformation is given by
\begin{equation}
\begin{aligned}
& \delta_{\Lambda} \Psi=Q_B \Lambda+\sum_{n=2}^{\infty} \frac{1}{n !} \ell_{n+1}^{\s{G}} \left(\Lambda, \Psi^{n}\right)~, \\
& \delta_{\Lambda} \widetilde{\Psi}=Q_B \widetilde{\Lambda}+\sum_{n=2}^{\infty} \frac{1}{n !} \ell_{n+1} \left(\Lambda,\Psi^{n}\right)~,
\end{aligned}
\end{equation}
where $|\Lambda\rangle \in \widehat{\mathcal{H}}_T$ and $|\widetilde{\Lambda}\rangle \in \widetilde{\mathcal{H}}_T$ both grassmann odd.
\section{Wilsonian effective action in superstring field theory}\label{sec:3}
In this section, we review the procedure of integrating out heavy fields to obtain the Wilsonian effective action as described in \cite{Sen:2016qap}. We then show that the effective action can be written equivalently in terms of certain effective string products. 
\subsection{The string field theory description}
Let us suppose we have a projection operator $P$ on a subset of string fields satisfying the conditions
\begin{equation}\label{eq:Pprop}
\left[P, b_0^{ \pm}\right]=0,\quad \left[P, c_0^{ \pm}\right]=0,\quad \left[P, L_0^{ \pm}\right]=0,\quad [P, \s{G}]=0,\quad \left[P, Q_B\right]=0 \text {. }
\end{equation}
We also assume that the projection operator is grassmann even and self-conjugate with respect to the BPZ inner product:
\begin{equation}\label{eq:Phermdef}
\left\langle\widetilde{\Psi}\left|c_0^{-} \right| P\Psi\right\rangle= \left\langle P\widetilde{\Psi}\left|c_0^{-} \right| \Psi\right\rangle,\quad \widetilde{\Psi}\in\widetilde{\mathcal{H}}_T,~~\Psi\in \widehat{\mathcal{H}}_T.    
\end{equation}
The usual way to get the Wilsonian effective action is by integrating out the states in $\bar{P}\left(\widetilde{\s{H}}_T \oplus \widehat{\s{H}}_T\right)$, where $\bar{P}=1-P$. To write down the effective action, we construct new string products $\{\ldots\}_e$ by using the propagator
\begin{equation}
\left\langle\varphi_s^c\left|c_0^{-} b_0^{+}\left(L_0^{+}\right)^{-1} \s{G} \bar{P}\right| \varphi_r^c\right\rangle \text {. }
\end{equation}
The effective action is then given by
\begin{equation}
S_e=\frac{1}{g_s^2}\left[-\frac{1}{2}\left\langle\widetilde{\psi}\left|c_0^{-} Q_B \s{G}\right| \widetilde{\psi}\right\rangle+\left\langle\widetilde{\psi}\left|c_0^{-} Q_B\right| \psi\right\rangle+\sum_{n=3}^{\infty} \frac{1}{n !}\left\{\psi^n\right\}_e\right]~,
\end{equation}
where $\psi \in P \widehat{\mathcal{H}}_T, \quad \widetilde{\psi} \in P \widetilde{\s{H}}_T$. \par We will show that this effective action can be obtained using homotopy transfer of the $\s{G}$-twisted cyclic $L_{\infty}$-superalgebra $\left(\wh{\s{H}}_T,\{\ell_n\},\s{G}\right)$. 
\subsection{Effective equation of motions}
Let us split the state $\widetilde{\Psi}, \Psi$ as
\begin{equation}\label{eq:decomppsiR}
\widetilde{\Psi}=\widetilde{\psi}+\widetilde{R}, \quad \Psi=\psi+R~,
\end{equation}
where $\widetilde{\psi}=P \widetilde{\Psi}, \psi=P \Psi, \widetilde{R}=\bar{P} \Psi, R=\bar{P} \Psi$. Next note that \eqref{eq:Phermdef} implies that 
\begin{equation}\label{eq:Pherm}
\omega\left(P \Psi, \Psi^{\prime}\right)=\omega\left(\Psi, P\Psi^{\prime}\right)~,
\end{equation}
since $P$ is grassmann even and commutes with $\s{G}$. Next define the degree-odd \footnote{It is degree-odd since $b_0^+$ is grassmann odd.} operator 
\begin{equation}
h=\frac{b_0^{+}}{L_0^{+}} \bar{P}.
\end{equation}
From the property \eqref{eq:Pprop} we have 
\begin{equation}\label{eq:hGcomm}
    [h,\s{G}]=0~, \quad \{h,c_0^-\}=0~,
\end{equation}
thus $h:\widetilde{\s{H}}_T\longrightarrow \widetilde{\s{H}}_T$ and also $h:\widehat{\s{H}}_T\longrightarrow \widehat{\s{H}}_T$.
Then we have the \textit{Hodge-Kodaira decomposition}
\begin{equation}\label{eq:HKdecomphQB}
\left\{h, Q_B\right\}=\bar{P} \text {. }
\end{equation}
Indeed using 
\begin{equation}
    \{AB,C\}=A[B,C]+\{A,C\}B~,
\end{equation}
we get
\begin{equation}
\begin{aligned}
\left\{\frac{b_0^{+}}{L_0^{+}} \bar{P}, Q_B\right\} & =\frac{\left\{b_0^{+}, Q_B\right\}}{L_0^{+}} \bar{P} \\
& =\frac{L_0^{+}}{L_0^{+}} \bar{P}=\bar{P}~.
\end{aligned}
\end{equation}
Also
\begin{equation}
\begin{aligned}
h^2 & =\left(b_0^{+}\left(L_0^{+}\right)^{-1} \bar{P}\right)\left(b_0^{+}\left(L_0^{+}\right)^{-1} \bar{P}\right) \\
& =b_0^{+}\left(L_0^{+}\right)^{-1} b_0^{+}\left(L_0^{+}\right)^{-1} \bar{P}^2 \\
& =\left(b_0^{+}\right)^2\left(L_0^{+}\right)^{-2} \bar{P}^2 \\
& =0 .
\end{aligned}
\end{equation}
We also have
\begin{equation}
\begin{aligned}
P h & =b_0^{+}\left(L_0^{+}\right)^{-1} P \bar{P}\\
& =0=h P .
\end{aligned}
\end{equation}
Thus we have the \textit{annihilation conditions}
\begin{equation}\label{eq:anncondhp}
    h^2=Ph=hP=0~.
\end{equation}
Substituting the decomposition \eqref{eq:decomppsiR} in the action \eqref{eq:clactellop} and varying separately with respect to $\widetilde{\psi}, \widetilde{R}, \psi, R$, we obtain
\begin{equation}
\begin{aligned}
& \operatorname{EOM}_{\widetilde{\psi}}(\widetilde{\psi}, \widetilde{R}, \psi)=P\operatorname{EOM}(\widetilde{\psi}+\widetilde{R})=Q_B(\psi-\s{G} \widetilde{\psi})~, \\
& \operatorname{EOM}_{\widetilde{R}}(\widetilde{\psi}, \widetilde{R}, R)=\bar{P} \operatorname{EOM}(\widetilde{\psi}+\widetilde{R})=Q_B(R-\s{G} \widetilde{R})~, \\
& \operatorname{EOM}_\psi(\psi, R, \widetilde{\psi})=P\operatorname{EOM}(\psi+R)=Q_B \widetilde{\psi}+P \s{J}^{\prime}(\psi+R)~, \\
& \operatorname{EOM}_R(\psi, R, \widetilde{R})=\bar{P} \operatorname{EOM}(\psi+R)=Q_B \widetilde{R}+\bar{P} \s{J}^{\prime}(\psi+R) .
\end{aligned}
\end{equation}
Again combining these EOMs we obtain the following two equations: 
\begin{equation}\label{eq:eompsiR}
\begin{aligned}
& \operatorname{EOM}_\psi(\psi, R)=P\operatorname{EOM}(\psi+R)=Q_B \psi+P \s{J}(\psi+R)~, \\
& \operatorname{EOM}_R(\psi, R)=\bar{P} \operatorname{POM}(\psi+R)=Q_B R+\bar{P} \s{J}(\psi+R) .
\end{aligned}
\end{equation}
\subsection{Fixing the Siegel gauge}
The Siegel gauge condition is
\begin{equation}
b_0^{+}|\widetilde{\Psi}\rangle=0,\quad b_0^{+}|\Psi\rangle=0 \implies b_0^{+}\left(|\Psi\rangle-\frac{1}{2} \s{G}|\widetilde{\Psi}\rangle\right)=0 \text {. }
\end{equation}
To integrate out $R$ and $\widetilde{R}$ from the dynamics, we need to use the EOM and solve for $R,\widetilde{R}$ and plug it back in the action. But before that we need to fix gauge for $\widetilde{R}, R$.
Define the Siegel gauge projector $P_s=b_0^{+} c_0^{+}$. Then 
\begin{equation}
\bar{P}_s \equiv 1-P_s=c_0^{+} b_0^{+},
\end{equation}
since 
\begin{equation}
    \left\{b_0^{+}, c_0^{+}\right\}=\left\{b_0^{-}, c_0^{-}\right\}=1. 
\end{equation}
Then decompose $\widetilde{R}=\widetilde{R}_1+\widetilde{R}_2, R=R_1+R_2$ where
\begin{equation}
\begin{aligned}
& \widetilde{R}_1=P_s \widetilde{R}, \quad \widetilde{R}_2=\bar{P}_s \widetilde{R}~, \\
& R_1=P_s R, \quad R_2=\bar{P}_s R ~.
\end{aligned}
\end{equation}
The Siegel gauge condition on $\widetilde{R}, R$ is then
\begin{equation}
\widetilde{R}_2=R_2=0.    
\end{equation} The EOM then decomposes into a set of equations in $\widetilde{R}_1, \widetilde{R}_2$ and $R_1, R_2$.
Next we have
\begin{equation}\label{eq:commPsP}
[P, P_s]=\left[P, b_0^{+} c_0^{+}\right]=0~.
\end{equation}
Thus using \eqref{eq:commPsP} and the EOM${}_R$ and EOM${}_{\widetilde{R}}$ we obtain 
\begin{equation}\label{eq:4eomsR1R1tilde}
\begin{aligned}
& \operatorname{EOM}_{\widetilde{R}_1}\left(\widetilde{\psi}, \widetilde{R}_1, R_1\right)=\bar{P}_s Q_B\left(R_1-\s{G} \widetilde{R}_1\right)~, \\
& \operatorname{EOM}_{R_1}\left(\psi, R_1, \widetilde{R}_1\right)=\bar{P}_s Q_B \widetilde{R}_1+\bar{P}_s \bar{P} \s{J}^{\prime}\left(\psi+R_1\right) ~, \\
& \operatorname{EOM}_{\widetilde{R}_2}\left(\widetilde{\psi}, \widetilde{R}_1, R_1\right)=P_s Q_B\left(R_1-\s{G} \widetilde{R}_1\right)~, \\
& \operatorname{EOM}_{R_2}\left(\psi, R_1, \widetilde{R}_1\right)=P_s Q_B \widetilde{R}_1+P_s \bar{P} \s{J}^{\prime}\left(\psi+R_1\right) .
\end{aligned}
\end{equation}
Again combining these equations we obtain
\begin{equation}\label{eq:eomR1_R1R1_R2}
\begin{aligned}
& \operatorname{EOM}_{R_1}\left(\psi, R_1\right)=\bar{P}_s Q_B R_1+\bar{P}_s \bar{P} \s{J}\left(\psi+R_1\right) \\
& \operatorname{EOM}_{R_2}\left(\psi, R_1\right)=P_s Q_B R_1+P_s \bar{P} \s{J}\left(\psi+R_1\right) .
\end{aligned}
\end{equation}
We now want to solve for $\widetilde{R}_1, R_1$ using the first equation above. Given a solution $R_1$ to the first equation of \eqref{eq:eomR1_R1R1_R2}, $\widetilde{R}_1$ is determined by the second equation of \eqref{eq:4eomsR1R1tilde} up to a $Q_B$-closed state.
We have
\begin{equation}
\bar{P}_s Q_B R_1+\bar{P}_s \bar{P} \s{J}\left(\psi+R_1\right)=0 .
\end{equation}
Since
\begin{equation}
\left\{b_0^{\pm}, Q_B\right\}=L_0^{ \pm} \text { and }\left(b_0^{+}\right)^2=0~,
\end{equation}
$\big($follows from $\left\{b_n, Q_B\right\}=L_n,\left\{\bar{b}_n, Q_B\right\}=\bar{L}_n$ and $\left\{b_n, \bar{b}_m\right\}=0\big),$ we have
\begin{equation}
\bar{P}_s Q_B R_1=c_0^{+} L_0^{+} R_1~,
\end{equation}
where we used the fact $b_0^{+} R_1=b_0^{+} P_s R=\left(b_0^{+}\right)^2 c_0^{+} R=0$. Thus $R_1$ must satisfy
\begin{equation}
\begin{aligned}
& c_0^{+} L_0^{+} R_1+c_0^{+} b_0^{+} \bar{P} J\left(\psi+R_1\right)=0 \\
\implies & c_0^{+} L_0^{+}\left(R_1+\left(L_0^{+}\right)^{-1} b_0^{+} \bar{P} \s{J}\left(\psi+R_1\right)\right)=0 \\
\implies & R_1(\psi)=-\left.h \s{J}(\Psi)\right|_{\Psi=\psi+R_1}.
\end{aligned}
\end{equation}
One can obtain recursive solution to this as follows: let $G=-h \circ \s{J}$. Then assuming $R_1(0)=0$, we get
\begin{equation}
R_1(\psi)=G(\psi+G(\psi+G(\psi+\ldots))) \text {. }
\end{equation}
More explicitly from $\Psi=\psi+R_1(\psi)$ we have
\begin{equation}\label{eq:Psirecsol}
\begin{aligned}
\Psi(\psi)= & \psi-\frac{1}{2 !} h \ell_2^{\s{G}}(\psi, \psi)-\frac{1}{3 !} h \ell_3^{\s{G}}(\psi, \psi, \psi)+\frac{2}{(2 !)^2} h \ell_2^{\s{G}}\left(h \ell_2^{\s{G}}(\psi, \psi), \psi\right)+ \\
& -\frac{1}{4 !} h \ell_4^{\s{G}}(\psi, \psi, \psi, \psi)+\frac{2}{2 ! 3 !} h \ell_2^{\s{G}}\left(h \ell_3^{\s{G}}(\psi, \psi, \psi), \psi\right)+\frac{3}{2 ! 3 !} h \ell_3^{\s{G}}\left(h \ell_2^{\s{G}}(\psi, \psi), \psi, \psi\right)+ \\
& -\frac{1}{(2 !)^3} h \ell_2^{\s{G}}\left(h \ell_2^{\s{G}}(\psi, \psi), h \ell_2^{\s{G}}(\psi, \psi)\right)-\frac{2^2}{(2 !)^3} h \ell_2^{\s{G}}\left(h \ell_2^{\s{G}}\left(h \ell_2^{\s{G}}(\psi, \psi), \psi\right), \psi\right)+O\left(\psi^5\right) .
\end{aligned}
\end{equation}
One can show that this solution satisfies the second equation of \eqref{eq:eomR1_R1R1_R2} as in \cite{Erbin:2020eyc}. It is easy to see that all the four equations in \eqref{eq:4eomsR1R1tilde} are satisfied  since \eqref{eq:eomR1_R1R1_R2} is equivalent to \eqref{eq:4eomsR1R1tilde}. Substituting the recursive solution \eqref{eq:Psirecsol} in the first equation of \eqref{eq:eompsiR} we obtain the eom for $\psi$:
\begin{equation}\label{eq:eompsi}
\operatorname{eom}(\psi)=\sum_{n=1}^{\infty} \frac{1}{n !} \widetilde{\ell}^{\s{G}}_n\left(\psi^{n}\right)
\end{equation}
where $\widetilde{\ell}^{\s{G}}_1=\widetilde{\ell}_1$ and  $\widetilde{\ell}^{\s{G}}_n=\s{G}\circ\widetilde{\ell}_n$ for $n>1$ where 
\begin{equation}\label{eq:lnGtildedef}
\begin{aligned}
& \widetilde{\ell}_1\left(\Psi_1\right)=P Q_B\Psi_1, \\
& \widetilde{\ell}_2\left(\Psi_1, \Psi_2\right)=P \ell_2\left(\Psi_1, \Psi_2\right), \\
& \widetilde{\ell}_3\left(\Psi_1, \Psi_2, \Psi_3\right)=P \ell_3\left(\Psi_1, \Psi_2, \Psi_3\right)-P \ell_2\left(\Psi_1, h \ell_2^{\s{G}}\left(\Psi_2, \Psi_3\right)\right)\\&\hspace{2.65cm}- (-1)^{d(\Psi_1)\left(d(\Psi_2)+d(\Psi_3)\right)} P \ell_2\left(\Psi_2, h \ell_2^{\s{G}}\left(\Psi_3, \Psi_1\right)\right) \\
&\hspace{2.65cm}-(-1)^{d(\Psi_3)\left(d(\Psi_1)+d(\Psi_2)\right)} P\ell_2\left(\Psi_3, h \ell_2^{\s{G}}\left(\Psi_1, \Psi_2\right)\right)~,
\end{aligned}
\end{equation}
and so on. 
In the next section, we will prove that $\left(P\widehat{\s{H}}_T,\{\widetilde{\ell}_n\}_{n=1}^\infty,\s{G}\right)$ is a $\s{G}$-twisted $L_{\infty}$-superalgebra. \par Notice that $h$ is BPZ self-conjugate: 
\begin{equation}\label{eq:hbpzconj}
\omega\left(\Psi_1, h \Psi_2\right)=(-1)^{d(\Psi_1)} \omega\left(h \Psi_1, \Psi_2\right),
\end{equation}
the factor of $(-1)^{d(\Psi_1)}$ appears because $h$ is grassmann odd. 
Let us prove that $\widetilde{\ell}_n$ are cyclic with respect to $\tilde{\omega}:=\omega|_{P\widehat{\s{H}}_T}$. 
We have for $\psi_i\in P\wh{\s{H}}$
\begin{equation}
\begin{aligned}
\tilde{\omega}\left(\psi_1, \wt{\ell}_2\left(\psi_2, \psi_3\right)\right) & =\tilde{\omega}\left(\psi_1, P \ell_2\left(\psi_2, \psi_3\right)\right) \\
& =\omega\left(P \psi_1, \ell_2\left(\psi_2, \psi_3\right)\right) \\
& =-(-1)^{d\left(\psi_1\right)} \omega\left(\ell_2\left(\psi_1, \psi_2\right), \psi_3\right) \\
& =-(-1)^{d\left(\psi_1\right)} \omega\left(\ell_2\left(\psi_1, \psi_2\right), P \psi_3\right) \\
& =-(-1)^{d\left(\psi_1\right)} \omega\left(P \ell_2\left(\psi_1, \psi_2\right), \psi_3\right) \\
& =-(-1)^{d\left(\psi_1\right)} \tilde{\omega}\left(\wt{\ell}_2\left(\psi_1, \psi_2\right), \psi_3\right),
\end{aligned}
\end{equation}
where we used the cyclicity of $\ell_n$ \eqref{eq:cycln}, $P^2=P$ and \eqref{eq:Pherm}. The same argument implies
\begin{equation}\label{eq:omPelln}
\begin{aligned}
\tilde{\omega}\left(\psi_1, P \ell \left(\psi_2, \ldots, \psi_{n+1}\right)\right) & =-(-1)^{d\left(\psi_1\right)} \tilde{\omega}\left(P \ell_n\left(\psi_1, \ldots, \psi_n\right), \psi_{n+1}\right) \\
& =-(-1)^{d\left(\psi_1\right)} \tilde{\omega}\left(\ell_n\left(\psi_1, \ldots, \psi_n\right), P \psi_{n+1}\right) \\
& =-(-1)^{d\left(\psi_1\right)} \tilde{\omega}\left(\ell_n\left(\psi_1, \ldots, \psi_n\right), \psi_{n+1}\right) .
\end{aligned}
\end{equation}
At next order, using \eqref{eq:omPelln} we have 
\begin{equation}
\begin{split}  \tilde{\omega}\left(\psi_1, \widetilde{\ell}_3\left(\psi_2, \psi_3, \psi_4\right)\right) &=-(-1)^{d\left(\psi_1\right)} \tilde{\omega}\left(P\ell_3\left(\psi_1, \psi_2, \psi_3\right), \psi_4\right) \\ &+(-1)^{d\left(\psi_1\right)} \tilde{\omega}\left(\ell_2\left(\psi_1, \psi_2\right), h \s{G} \ell_2\left(\psi_3, \psi_4\right)\right) \\ &+(-1)^{d\left(\psi_1\right)}(-1)^{d\left(\psi_2\right)d\left(\psi_3\right)} \tilde{\omega}\left(\ell_2\left(\psi_1, \psi_3\right), h \s{G} \ell_2\left(\psi_2, \psi_4\right)\right) \\ &+(-1)^{d\left(\psi_1\right)}\tilde{\omega}\left(\ell_2\left(\psi_1 ,h  \ell_2^{\s{G}}\left(\psi_2, \psi_3\right)\right), \psi_4\right)\\&=-(-1)^{d\left(\psi_1\right)}\big[\tilde{\omega}\left(P\ell_3\left(\psi_1, \psi_2, \psi_3\right), \psi_4\right) -\tilde{\omega}\left(\ell_2\left(h\ell_2^{ \s{G} }\left(\psi_1, \psi_2\right), \psi_3\right), \psi_4\right) \\&\hspace{2cm}-(-1)^{d\left(\psi_2\right)d\left(\psi_3\right)} \tilde{\omega} \left(\ell_2\left(h \ell_2^{\s{G}}\left(\psi_1, \psi_3\right), \psi_2\right), \psi_4\right) \\&\hspace{2cm}-\tilde{\omega}\left(\ell_2\left(\psi_1 ,h  \ell_2^{\s{G}}\left(\psi_2, \psi_3\right)\right), \psi_4\right)\big]\\&=-(-1)^{d\left(\psi_1\right)}\tilde{\omega}\left(\widetilde{\ell}_3\left(\psi_1, \psi_2, \psi_3\right), \psi_4\right) ~,
\end{split}    
\end{equation}
where we used the fact that $h,\ell_2$ are both degree odd with respect $d(\cdot)$ and graded-antisymmetry \eqref{eq:grantsymln} of $\ell_n$ in the first step and the cyclicity \eqref{eq:cycln} of $\ell_n$, its graded antisymmetry, $[h,\s{G}]=0$ and \eqref{eq:hbpzconj} in the second and third step. 
\\\\
As before, we can use \eqref{eq:4eomsR1R1tilde} and the solution $\psi$ to determine $\widetilde{\psi}$ up to a $Q_B$-exact field. The effective action which captures the dynamics of $\psi$ determined by the eom \eqref{eq:eompsi} of $\psi$ can be written as
\begin{equation}\label{eq:steffaction}
\begin{split}
S^{\text{eff}}\left(\psi, \widetilde{\psi}\right)&=\frac{1}{g_s^2}\left[-\frac{1}{2}\tilde{\omega}\left(\widetilde{\psi},Q_B\s{G}\widetilde{\psi}\right)+\tilde{\omega}\left(\widetilde{\psi},Q_B\psi\right)+\sum_{n=1}^\infty\frac{1}{(n+1)!}\tilde{\omega}\left(\psi,\widetilde{\ell}_n(\psi^n)\right)\right]\\&=\frac{1}{g_s^2}\left[-\frac{1}{2}\tilde{\omega}(\widetilde{\psi},Q_B\s{G}\widetilde{\Psi})+\tilde{\omega}(\widetilde{\Psi},Q_B\psi)+\sum_{n=1}^\infty\frac{1}{n!} \int_0^1 dt~\tilde{\omega}(\dot{\psi}(t),\ell_n(\psi(t)^n))\right]~,
\end{split}
\end{equation}
where $\psi(t)$ is a smooth interpolation for $0\leq t\leq 1$ with $\psi(0)=0$ and $\psi(1)=\psi$.
Again $\widetilde{\psi}\in P\widetilde{\s{H}}_T$ is a non-dynamical field. It can easily be checked that this effective action has an infinite dimensional gauge symmetry given by 
\begin{equation}
\begin{aligned}
& \delta_{\lambda} \psi=Q_B \lambda+\sum_{n=2}^{\infty} \frac{1}{n !} \wt{\ell}_{n+1}^{\s{G}} \left(\lambda, \psi^{n}\right)~, \\
& \delta_{\lambda} \widetilde{\psi}=Q_B \widetilde{\lambda}+\sum_{n=2}^{\infty} \frac{1}{n !} \wt{\ell}_{n+1} \left(\lambda,\psi^{n}\right)~,
\end{aligned}
\end{equation}
where $|\lambda\rangle \in P\widehat{\mathcal{H}}_T$ and $|\widetilde{\lambda}\rangle \in P\widetilde{\mathcal{H}}_T$ both grassmann odd.
\\\\
We can capture the system in terms of a strong deformation retract (SDR), see Appendix \ref{app:sdr} for definitions. We have the projection operator
\begin{equation}\label{eq:Pimapdef}
\Pi: \widehat{\mathcal{H}}_T \longrightarrow P \widehat{\mathcal{H}}_T,
\end{equation}
and the natural inclusion map
\begin{equation}\label{eq:Imapdef}
    \mathrm{I}: P\widehat{\mathcal{H}}_T \longrightarrow 
      \widehat{\mathcal{H}}_T.
\end{equation}
The Hodge-Kodaira decomposition \eqref{eq:HKdecomphQB} and the annihilation conditions \eqref{eq:anncondhp} imply that we have a strong deformation retract (SDR)
\begin{equation}\label{eq:sdrHhatHtilde}
\rotatebox[origin=c]{270}{$\mathlarger{\mathlarger{\mathlarger{\mathlarger{\circlearrowright}}}}$}~(-h)\left(\widehat{\mathcal{H}}_T, Q_B\right) \underset{\mr{I}}{\stackrel{\Pi}{\myrightleftarrows{\rule{2cm}{0cm}}}}\left(P \widehat{\mathcal{H}}_T, \Pi Q_B\mr{I} \right),
\end{equation}
where the propagator $h$ is (minus) the homotopy operator, see Appendix \ref{app:sdr} for definitions. 
We also have maps 
\begin{equation}
\Pi: \widetilde{\mathcal{H}}_T \longrightarrow P \widetilde{\mathcal{H}}_T,\quad   \mathrm{I}: P\widetilde{\mathcal{H}}_T \longrightarrow 
      \widetilde{\mathcal{H}}_T,\quad h: \widetilde{\mathcal{H}}_T \longrightarrow 
      \widetilde{\mathcal{H}}_T.   
\end{equation}
This means that we have another SDR
\begin{equation}
\begin{split}
\rotatebox[origin=c]{270}{$\mathlarger{\mathlarger{\mathlarger{\mathlarger{\circlearrowright}}}}$}~(-h)\left(\widetilde{\s{H}}_T, Q_B\right) \underset{\mr{I}}{\stackrel{\Pi}{\myrightleftarrows{\rule{2cm}{0cm}}}}\left(P \widetilde{\s{H}}_T, \Pi Q_B \mr{I}\right)~.
\end{split}
\end{equation}
\section{Effective action from homotopy transfer}\label{sec:4}
In this section, we will prove that $\left(P \widehat{\s{H}}_T,\{\widetilde{\ell}_n\}_{n=1}^{\infty},\s{G}\right)$ is a $\s{G}$-twisted $L_{\infty}$-superalgebra. We will prove this using the homological perturbation lemma and symmetrized tensor coalgebra language. See Appendix \ref{app:coalg} and Appendix \ref{app:sdr} for discussion on these topics. 
\subsection{Closed superstring field theory as symmetrized tensor coalgebra}
We begin by describing the closed superstring field theory in terms of symmetrized tensor coalgebra, see Appendix \ref{app:coalg} for review of definitions.  
Let $S\widetilde{\s{H}}_{T},S\widehat{\s{H}}_{T}$ be the symmetric algebra given by 
\begin{equation}
\begin{split}
&S\widetilde{\s{H}}_{T}=\widetilde{\s{H}}_{T}^{\wedge 0} \oplus \widetilde{\s{H}}_{T}^{\wedge 1} \oplus \widetilde{\s{H}}_{T}^{\wedge 2} \oplus \ldots\\& S\widehat{\s{H}}_{T}=\widehat{\s{H}}_{T}^{\wedge 0} \oplus \widehat{\s{H}}_{T}^{\wedge 1} \oplus \widehat{\s{H}}_{T}^{\wedge 2} \oplus \ldots~,
\end{split}
\end{equation}
where $\widetilde{\s{H}}_{T}^{\wedge 0}=\C\textbf{1}_{S\widetilde{\mathcal{H}}_T},~\widehat{\s{H}}_{T}^{\wedge 0}=\C\textbf{1}_{S\widehat{\mathcal{H}}_T}$ and $\textbf{1}_{S\widetilde{\mathcal{H}}_T},\textbf{1}_{S\widehat{\mathcal{H}}_T}$ is the identity of the coalgebra $S\widetilde{\mathcal{H}}_T,S\widehat{\mathcal{H}}_T$ respectively. As described in Appendix \ref{app:coalg}, these are coalgebras with the counit and coproduct given by \eqref{eq:counitTV} and \eqref{eq:coprodSV}. \par Let us now express the SDR \eqref{eq:sdrHhatHtilde} in terms of tensor coalgebra. Let $\bm{Q_B}$ denote the coderivation on $S\widehat{\s{H}}_T$ and $S\widetilde{\s{H}}_T$ obtained from \eqref{eq:coderbmci} using the BRST charge $Q_B$:
\begin{equation}
\bm{Q_B}|_{S\wh{\s{H}}}=\sum_{n\geq 1}\left(Q_B\wedge\mathds{1}_{\wh{\s{H}}^{\wedge n-1}}\right)\pi_n,\quad \bm{Q_B}|_{S\wt{\s{H}}}=\sum_{n\geq 1}\left(Q_B\wedge\mathds{1}_{\wt{\s{H}}^{\wedge n-1}}\right)\pi_n~,   
\end{equation}
where $\pi_n$ is the projection onto the $n^{th}$ homogenous subspace. 
We will simply denote it by $\bm{Q_B}$, the space it acts on will be clear from context. 
The projection $P$ and the maps in \eqref{eq:Pimapdef}, \eqref{eq:Imapdef} gives rise to linear maps on 
$S\widehat{\s{H}}_T$ and $S\widetilde{\s{H}}_T$. Define 
\begin{equation}
\mathbf{P}: S\widehat{\s{H}}_T \longrightarrow S\widehat{\s{H}}_T, \quad \mathbf{I}: SP\widehat{\s{H}}_T\longrightarrow S\widehat{\s{H}}_T,\quad \bm{\Pi}: S\widehat{\s{H}}_T \longrightarrow SP\widehat{\s{H}}_T,
\end{equation}
where $SP\widehat{\s{H}}_T$ is the symmetrized tensor coalgebra on the image $P\widehat{\s{H}}_T$ of $P$ in $\widehat{\s{H}}_T$,
by
\begin{equation}
\mathbf{P}=\sum_{k=0}^\infty \mathbf{P}_k,\quad    \mathbf{I}=\sum_{k=0}^\infty \mathbf{I}_k,\quad \bm{\Pi}=\sum_{k=0}^\infty \bm{\Pi}_k~,
\end{equation}
where 
\begin{equation}
\begin{aligned}
&\mathbf{P}_k:= \mathbf{P} \pi_k:=\frac{1}{k !} P^{\wedge k} \pi_k~, \\
& \mathbf{I}_k:= \mathbf{I} \pi_k:=\frac{1}{k !} \mr{I}^{\wedge k} \pi_k~, \\
& \bm{\Pi}_k:= \bm{\Pi} \pi_k:=\frac{1}{k !} \Pi^{\wedge k} \pi_k.
\end{aligned}
\end{equation}
Similarly we get linear maps $\mathbf{P},\bm{\Pi},\textbf{I}$ on $S\widetilde{\s{H}}_T$ corresponding to $P,\Pi,\mr{I}$ respectively. It is easy to check that 
\begin{equation}
    \mathbf{P}=\mathbf{I}\bm{\Pi},\quad \bm{\Pi}\mathbf{I}=\1,
\end{equation}
where $\1$ denotes the identity operator on $SP\widehat{\s{H}}_T$ or $SP\widetilde{\s{H}}_T$.
We also want to define the propagator $\bm{h}: S\widehat{\s{H}}_T \longrightarrow S\widehat{\s{H}}_T$ and $\bm{h}: S\widetilde{\s{H}}_T \longrightarrow S\widetilde{\s{H}}_T$ which satisfies the Hodge-Kodaira decomposition as operators on the symmetrized tensor coalgebra:
\begin{equation}
\bm{Q_B h}+\bm{hQ_B}=\mathds{1}_{S\widehat{\s{H}}_T}-\mathbf{P}~.
\end{equation}
The following choice does the trick \cite{Erbin:2020eyc} 
\begin{equation}
    \bm{h}:=\sum_{k=0}^\infty \bm{h}_k,
\end{equation}
where 
\begin{equation}
\bm{h}_k:=\bm{h} \pi_k:=\frac{1}{k !} \sum_{j=0}^{k-1}\left[h \wedge P^{\wedge j} \wedge\left(\mathds{1}_{\widehat{\s{H}}_T}\right)^{\wedge(k-1-j)}\right] \pi_k.
\end{equation}
With this choice, the annihilation conditions are satisfied:
\begin{equation}\label{eq:anncond}
    \bm{h}\mathbf{I}=\mathbf{P} \bm{h}=\bm{h}^2=0.
\end{equation}
Thus we have an SDR of the free theory
\begin{equation}\label{eq:sdrSHhatG}
\rotatebox[origin=c]{270}{$\mathlarger{\mathlarger{\mathlarger{\mathlarger{\circlearrowright}}}}$}~(-\bm{h})\left(S\widehat{\mathcal{H}}_T, \bm{Q_B}\right) \underset{\mathbf{I}}{\stackrel{\bm{\Pi}}{\myrightleftarrows{\rule{2cm}{0cm}}}}\left(SP \widehat{\mathcal{H}}_T, \bm{\Pi Q_B} \mathbf{I}\right).
\end{equation}
As described in Appendix \ref{app:coalg} (Theorem \ref{thm:codtoend}), the string products $\ell_n^{\s{G}}: \wh{\mathcal{H}}_T^{\wedge n} \longrightarrow \widehat{\s{H}}_T$ for $n\geq 1$ with $\ell_1^{\s{G}}=Q_B$ gives us a coderivation
\begin{equation}
    \bm{\ell^{\s{G}}}=\mu\circ (\ell^{\s{G}}\boxtimes\1_{S\widehat{\s{H}}_{T}})\circ\Delta_{S\widehat{\s{H}}_{T}}:S\widehat{\s{H}}_{T}\longrightarrow S\widehat{\s{H}}_{T},\quad \ell^{\s{G}}=\sum_{n=1}^\infty \ell^{\s{G}}_n~,
\end{equation} 
where $\mu:S\widehat{\s{H}}_{T}\boxtimes S\widehat{\s{H}}_{T}\longrightarrow S\widehat{\s{H}}_{T}$ is the multiplication map and $\Delta_{S\widehat{\s{H}}_{T}}:S\widehat{\s{H}}_{T}\longrightarrow S\widehat{\s{H}}_{T}\boxtimes S\widehat{\s{H}}_{T}$ is the coproduct. 
Then the strong homotopy Jacobi identity is equivalent to the vanishing of the commutator 
\begin{equation}\label{eq:bmelg2=0}
(\bm{\ell^{\s{G}}})^2=\frac{1}{2}[\bm{\ell^{\s{G}}},\bm{\ell^{\s{G}}}]=0,   
\end{equation}
\par
For an even-degree vector $A$, introduce the notation 
\begin{equation}
    e^{\wedge A}:=\sum_{n=0}^\infty\frac{A^{\wedge n}}{n!}=\textbf{1}_{S\widehat{\mathcal{H}}_T}+A+\frac{A\wedge A}{2!}+\frac{A\wedge A\wedge A}{3!}+\dots
\end{equation}
To write the action \eqref{eq:stactionint} in the coalgebra langauge, we introduce the bra action of $\omega$ as:
\begin{equation}
\begin{split}
&\langle\omega|:\wh{\s{H}}_T\otimes \wt{\s{H}}_T\longrightarrow \C,\\&\langle\omega|\wt{\Psi}\otimes\Psi=\omega(\wt{\Psi},\Psi)~.
\end{split}
\end{equation}
Similarly $\langle\omega|:\wt{\s{H}}_T\otimes \wh{\s{H}}_T\longrightarrow \C$ is also defined. 
The action \eqref{eq:stactionint} can then be written as 
\begin{equation}\label{eq:stactioncoalg}
S_{\text{cl}}=\frac{1}{g_s^2}\left[-\frac{1}{2}\langle\omega|\widetilde{\Psi}\otimes Q_B\s{G}\widetilde{\Psi}+\langle\omega|\widetilde{\Psi}\otimes Q_B\Psi+\int_0^1 dt~\langle\omega|\pi_1\bm{\partial_t}e^{\wedge \Psi(t)}\otimes \mu\circ(\pi_1\boxtimes\pi_0)\circ\delta\bm{\ell}e^{\wedge \Psi(t)}\right]~,   
\end{equation}
where $\bm{\partial_t}$ is the coderivation corresponding to the linear map $\partial_t$.
Then the equation of motion obtained from this action takes the form
\begin{equation}
\begin{aligned}
& \operatorname{EOM}(\widetilde{\Psi})=Q_B(\Psi-\s{G} \widetilde{\Psi})~, \\
& \operatorname{EOM}(\Psi)=Q_B \widetilde{\Psi}+\mu\circ(\pi_1\boxtimes\pi_0)\circ\delta\bm{\ell}e^{\wedge \Psi(t)}~,
\end{aligned}
\end{equation}
which can be combined to give 
\begin{equation}
    \pi_1\bm{\ell^{\s{G}}}e^{\wedge \Psi}=0.
\end{equation}
This equation is identical to the EOM \eqref{eq:EOMPsi}.
\subsection{Perturbed SDR and effective string products}
We will now use the homological perturbation lemma to deduce that the effective string products $\wt{\ell}_n$ in \eqref{eq:lnGtildedef} satisfy the strong homotopy Jacobi identity. Introduce the perturbation 
\begin{equation}
\begin{split}
    \delta\bm{\ell}:&S\widehat{\mathcal{H}}_T\longrightarrow S\widetilde{\mathcal{H}}_T\boxtimes S\widehat{\mathcal{H}}_T\\
&\delta \bm{\ell}=\sum_{k>1} (\ell_k\boxtimes\1_{S\widehat{\mathcal{H}}_T})\circ\Delta_{S\widehat{\mathcal{H}}_T}~, 
\end{split}
\end{equation}
to the free SDR \eqref{eq:sdrSHhatG}. 
The PCO $\s{G}$ gives rise to a linear map 
\begin{equation}
\bm{\s{G}}:S\widetilde{\s{H}}_{T} \longrightarrow S\widehat{\s{H}}_{T}~,
\end{equation} as follows:
\begin{equation}\label{eq:bmgonSV}
    \bm{\s{G}}=\sum_{k=0}^\infty \bm{\s{G}}_k,\quad \bm{\s{G}}_k:= \bm{\s{G}} \pi_k:=\frac{1}{k !} \s{G}^{\wedge k} \pi_k~.
\end{equation}
Then from \eqref{eq:bmelg2=0}, the map 
\begin{equation}
\mu\circ(\bm{\s{G}}\boxtimes\1_{S\widehat{\s{H}}_{T}}):S\widetilde{\s{H}}_{T}\boxtimes S\widehat{\s{H}}_{T} \longrightarrow S\widehat{\s{H}}_{T}~,    
\end{equation}
satisfies 
\begin{equation}
    (\bm{Q_B}+\delta \bm{\ell^{\s{G}}})^2=(\bm{\ell^{\s{G}}})^2=0~,
\end{equation}
where 
\begin{equation}
 \delta \bm{\ell^{\s{G}}}:=\mu\circ(\bm{\s{G}}\boxtimes\1_{S\widehat{\s{H}}_{T}})\circ\delta\bm{\ell}~.   
\end{equation}
Moreover, since $\s{G}$ commutes with $P$, we also have the restriction map 
\begin{equation}
\bm{\s{G}}:SP\widetilde{\s{H}}_{T} \longrightarrow SP\widehat{\s{H}}_{T}~,    
\end{equation}
and it satisfies  
\begin{equation}
\mu\circ(\bm{\s{G}}\boxtimes\1_{SP\widehat{\s{H}}_{T}})\circ(\bm{\Pi}\boxtimes\bm{\Pi})=\bm{\Pi}\circ \mu\circ(\bm{\s{G}}\boxtimes\1_{S\widehat{\s{H}}_{T}})~.    
\end{equation}
Thus by homological perturbation Theorem \ref{thm:hpt}, we obtain the perturbed SDR:
\begin{equation}\label{eq:pertsdrSHhat}
\rotatebox[origin=c]{270}{$\mathlarger{\mathlarger{\mathlarger{\mathlarger{\circlearrowright}}}}$}~(-\widetilde{\bm{h}})\left(S\widehat{\mathcal{H}}_T, \bm{\ell^{\s{G}}}\right) \underset{\widetilde{\mathbf{I}}}{\stackrel{\widetilde{\bm{\Pi}}}{\myrightleftarrows{\rule{2cm}{0cm}}}}\left(SP \widehat{\mathcal{H}}_T, \bm{\Pi Q_B} \mathbf{I}+\delta\widetilde{\bm{\ell^{\s{G}}}}\right),
\end{equation}
where
\begin{equation}
\begin{aligned}
& \delta\widetilde{\bm{\ell}}=(\bm{\Pi}\boxtimes\bm{\Pi}) \delta \bm{\ell} \frac{1}{\mathds{1}_{S\widehat{\mathcal{H}}_T}+\bm{h} \delta\bm{\ell^{\s{G}}}} \mathbf{I}, \\
& \widetilde{\bm{h}}=\bm{h}-\bm{h}\delta\bm{\ell^{\s{G}}}\frac{1}{\mathds{1}_{S\widehat{\mathcal{H}}_T}+\bm{h} \delta\bm{\ell^{\s{G}}}} \bm{h}, \\
& \widetilde{\mathbf{I}}=\mathbf{I}-\bm{h}\delta\bm{\ell^{\s{G}}}\frac{1}{\mathds{1}_{S \widehat{\mathcal{H}}_T}+\bm{h} \delta\bm{\ell^{\s{G}}}} \mathbf{I}, \\
& \widetilde{\bm{\Pi}}=\bm{\Pi}-\bm{\Pi}\delta\bm{\ell^{\s{G}}} \frac{1}{\mathds{1}_{S \widehat{\mathcal{H}}_T}+ \bm{h}\delta\bm{\ell^{\s{G}}}}\bm{h}.
\end{aligned}
\end{equation}
and 
\begin{equation}
\delta\widetilde{\bm{\ell^{\s{G}}}}=\mu\circ(\bm{\s{G}}\boxtimes\1_{SP\widehat{\s{H}}_{T}})\circ\delta\widetilde{\bm{\ell}}.     
\end{equation}
Using the definition
\begin{equation}
    \frac{1}{1+A}=\sum_{n=0}^{\infty}(-A)^n,
\end{equation}
we get 
\begin{equation}
\begin{aligned}
\widetilde{\bm{h}} & =\frac{1}{\mathds{1}_{S \widehat{\mathcal{H}}_T}+\bm{h} \delta \bm{\ell^{\s{G}}}} \bm{h}, \\
\widetilde{\mathbf{I}} & =\frac{1}{\mathds{1}_{S \widehat{\mathcal{H}}_T}+\bm{h} \delta \bm{\ell^{\s{G}}}} \mathbf{I}, \\
\widetilde{\bm{\Pi}} & =\bm{\Pi} \frac{1}{\mathds{1}_{S \widehat{\mathcal{H}}_T}+\delta \bm{\ell^{\s{G}}}\bm{h}}.
\end{aligned}
\end{equation}
Since \eqref{eq:pertsdrSHhat} is an SDR, we have
\begin{equation}\label{eq:lGtilde2=0}
\left(\widetilde{\bm{\ell^{\s{G}}}}\right)^2=\frac{1}{2}\left[\widetilde{\bm{\ell^{\s{G}}}},\widetilde{\bm{\ell^{\s{G}}}}\right]=0, 
\end{equation}
where 
$\widetilde{\bm{\ell^{\s{G}}}}=\bm{\Pi Q_B} \mathbf{I}+\delta\widetilde{\bm{\ell^{\s{G}}}}$.
The effective products are given by 
\begin{equation}
   \wt{\ell}_1=PQ_B,\quad \wt{\ell}_n:=\mu\circ(\pi_1\boxtimes\pi_0)\circ\delta\widetilde{\bm{\ell}}\circ\pi_n,\quad n>1~,
\end{equation}
where the image $\wh{\s{H}}^{\wedge 0}\cong\C$ of $\pi_0$ is considered as the space $\wt{\s{H}}^{\wedge 0}\cong\C$.
\eqref{eq:lGtilde2=0} then implies that the effective products $\wt{\ell}_n$ satisfies the strong homotopy Jacobi identity. It is also straightforward to check that $\wt{\ell}_n$ defined above agrees with \eqref{eq:lnGtildedef}. 
As in \cite[Appendix B.2.3]{Erbin:2020eyc}, we have
\begin{equation}
\begin{split}
\Psi(\psi)&=\pi_1\widetilde{\mathbf{I}}e^{\wedge\psi}=\pi_1\frac{1}{\mathds{1}_{S \widehat{\mathcal{H}}_T}+\bm{h} \delta\bm{\ell^{\s{G}}}} \mathbf{I}e^{\wedge\psi}   \\& =  \psi-\frac{1}{2 !} h \ell_2^{\s{G}}(\psi, \psi)-\frac{1}{3 !} h \ell_3^{\s{G}}(\psi, \psi, \psi)+\frac{2}{(2 !)^2} h \ell_2^{\s{G}}\left(h \ell_2^{\s{G}}(\psi, \psi), \psi\right)+ \\
& -\frac{1}{4 !} h \ell_4^{\s{G}}(\psi, \psi, \psi, \psi)+\frac{2}{2 ! 3 !} h \ell_2^{\s{G}}\left(h \ell_3^{\s{G}}(\psi, \psi, \psi), \psi\right)+\frac{3}{2 ! 3 !} h \ell_3^{\s{G}}\left(h \ell_2^{\s{G}}(\psi, \psi), \psi, \psi\right)+ \\
& -\frac{1}{(2 !)^3} h \ell_2^{\s{G}}\left(h \ell_2^{\s{G}}(\psi, \psi), h \ell_2^{\s{G}}(\psi, \psi)\right)-\frac{2^2}{(2 !)^3} h \ell_2^{\s{G}}\left(h \ell_2^{\s{G}}\left(h \ell_2^{\s{G}}(\psi, \psi), \psi\right), \psi\right)+O\left(\psi^5\right) .
\end{split}    
\end{equation}
Thus we have shown that $\left(P \widehat{\s{H}}_T,\{\widetilde{\ell}_n\}_{n=1}^{\infty},\s{G}\right)$ is a $\s{G}$-twisted $L_{\infty}$-superalgebra.\\\\
The maps $\widetilde{\mathbf{I}},\widetilde{\bm{\Pi}}$ can be shown to be cohomomorphisms provided that we can prove that  
\begin{equation}\label{eq:cohomIPimainid}
\Delta_{S \wh{\mathcal{H}}}(\bm{h} \delta \boldsymbol{\ell^{\s{G}}})^k \mathbf{I}=\sum_{n=0}^k\left[(\bm{h} \delta \boldsymbol{\ell^{\s{G}}})^n \mathbf{I} \boxtimes(\bm{h} \delta \boldsymbol{\ell^{\s{G}}})^{k-n} \mathbf{I}\right] \Delta_{S P \wh{\mathcal{H}}}.
\end{equation}
An order-by-order proof of \eqref{eq:cohomIPimainid} appeared in \cite[Appendix B]{Erbin:2020eyc}. Since $\bm{\s{G}}$ commutes with $\mathbf{\Pi}$, using \eqref{eq:cohomIPimainid}, one can prove that $\delta\widetilde{\bm{\ell^{\s{G}}}}$ and hence $\widetilde{\bm{\ell^{\s{G}}}}$ is a coderivation. The effective action can be written using the effective string product as (cf. \eqref{eq:stactioncoalg}):
\begin{equation}
S^{\text{eff}}=\frac{1}{g_s^2}\left[-\frac{1}{2}\langle\tilde{\omega}|\widetilde{\psi}\otimes Q_B\s{G}\widetilde{\psi}+\langle\tilde{\omega}|\widetilde{\psi}\otimes Q_B\psi+\int_0^1 dt~\langle\tilde{\omega}|\pi_1\bm{\partial_t}e^{\wedge \psi(t)}\otimes \mu\circ(\pi_1\boxtimes\pi_0)\circ\delta\wt{\bm{\ell}}e^{\wedge \psi(t)}\right]~.       
\end{equation}
Unpacking this reproduces the action \eqref{eq:steffaction}.
\\\\
\textbf{Acknowledgements.} RKS would like to thank Ashoke Sen for suggesting this problem and numerous discussions on superstring field theory, Carlo Maccaferri for comments which improved the presentation of the paper and the anonymous referee for suggestions which improved certain aspects of the paper. RKS thanks ICTS, Bangalore for hospitality where this work began. The work of RKS was supported by the US Department of Energy under grant DE-SC0010008.
\begin{appendix}
\section{(Symmetrized) Tensor coalgebra}\label{app:coalg}
We begin by describing the general construction of the symmetrized tensor coalgebra, the reader is referred to \cite{bourbaki1998algebra} for more details. 
\par Throughout this section $\F$ is field of characteristic $\neq 2$. 
\begin{defn}
A \textit{coalgebra} over a field $\F$ is an $\F$-vector space $C$ along with $\F$-linear maps 
\begin{equation}
\Delta:C\longrightarrow C\otimes C,\quad \varepsilon: C\longrightarrow \F~,    
\end{equation}
called \textit{coproduct} and \textit{counit} respectively, 
such that 
\begin{equation}\label{eq:coasscouprop}
\begin{split}
&(\mathds{1}_C\otimes\Delta)\circ\Delta=(\Delta\otimes\mathds{1}_C)\circ\Delta~,\\&(\mathds{1}_C\otimes\varepsilon)\circ\Delta=\mathds{1}_C=(\varepsilon\otimes\mathds{1}_C)\circ\Delta.
\end{split}    
\end{equation}
The first and second relations in \eqref{eq:coasscouprop} are called \textit{coassociativity} and \textit{counit property} respectively. 
\end{defn}
Let $(C,\Delta,\varepsilon)$ and $(C',\Delta',\varepsilon')$ be two coalgebras. Then a \textit{cohomomorphism} between $(C,\Delta,\varepsilon)$ and $(C',\Delta',\varepsilon')$ is an $\F$-linear map $f:C\longrightarrow C'$ such that 
\begin{equation}
    \Delta'\circ f=(f\otimes f)\circ\Delta,\quad \varepsilon'\circ f=\varepsilon,
\end{equation}
that is, the following diagram commutes: 
\[\begin{tikzcd}
	C && {C'} && C && {C'} \\
	{C\otimes C} && {C'\otimes C'} &&& \F
	\arrow[from=1-1, to=1-3]
	\arrow["\Delta"', from=1-1, to=2-1]
	\arrow[from=2-1, to=2-3]
	\arrow["{\Delta'}", from=1-3, to=2-3]
	\arrow["f", from=1-5, to=1-7]
	\arrow["\varepsilon"', from=1-5, to=2-6]
	\arrow["{\varepsilon'}", from=1-7, to=2-6]
\end{tikzcd}\]
The identity map $\mathds{1}_C$ is clearly a cohomomorphism. 
An $\F$-linear endomorphism $d:C\longrightarrow C$ is called a \textit{coderivation} if it satisfies the \textit{co-Leibnitz rule}
\begin{equation}
    \Delta\circ d=(d\otimes\mathds{1}_C+\mathds{1}_C\otimes d)\circ \Delta.
\end{equation}
The space of all coderivations forms a vector space which we denote by $\text{Coder}(V)$. 
An element $c\in C$ is called a \textit{group-like element} if 
\begin{equation}
    \Delta(c)=c\otimes c,\quad \varepsilon(c)=1.
\end{equation}
The canonical example of coalgebra associated to a vector space is the \textit{tensor coalgebra}. Let $V$ be an $\F$-vector space and let $TV$ denote the tensor algebra 
\begin{equation}
TV=\bigoplus_{n=0}^\infty V^{\otimes n},     
\end{equation}
where $V^{\otimes 0}=\F\mathbf{1}_{TV}\cong\F$ and $\mathbf{1}_{TV}$ should be considered as the identity of the  tensor algebra. The coproduct is given by 
\begin{equation}
    \begin{split}
    \Delta_{TV}:TV&\longrightarrow TV\boxtimes TV\\v_1\otimes v_2\otimes\dots\otimes v_n\mapsto \Delta_{TV}(v_1\otimes v_2\otimes\dots\otimes v_n)&:=\sum_{k=0}^n ( v_1\otimes\dots\otimes v_k)\boxtimes (v_{k+1}\otimes \dots\otimes v_{n})
    \end{split}
\end{equation}
where we have used $\boxtimes$ to denote an ``external'' tensor product required to define the coproduct and differentiate it from the ``internal'' tensor product already used in defining the $TV$. The counit is given by 
\begin{equation}\label{eq:counitTV}
\begin{split}
    &\varepsilon:TV\longrightarrow \F,\\ a\mapsto a\textbf{1}_{TV},~~v\mapsto 0,\quad &a\textbf{1}_{TV}\in T^0V=\F\textbf{1}_{TV},\quad v\in\bigoplus_{n=1}^\infty T^nV.
\end{split}
\end{equation}
This notion of tensor coalgebra is useful to describe the open string field theory \cite{Gaberdiel:1997ia,Kajiura:2003ax,Kajiura:2004xu,Erbin:2020eyc}.\\\\ We now come to the main example relevant for our purposes. We will assume that $V$ is $\Z$-graded. 
Let $SV$ denote the symmetrized tensor algebra defined by (see Subsection \ref{sec:genLinfalg} for notations)
\begin{equation}
SV=\bigoplus_{n=0}^\infty V^{\wedge n}.    
\end{equation}
Here $V^{\wedge 0}=\F\mathbf{1}_{SV}\cong\F$. Here as before $\mathbf{1}_{SV}$ should be considered as the identity of the symmetric tensor algebra. The coproduct is given by 
\begin{equation}\label{eq:coprodSV}
\begin{split}
\Delta_{SV}:&SV\longrightarrow SV\boxtimes SV~,\\
\Delta_{SV}\left(v_1 \wedge \ldots \wedge v_n\right)&=\sum_{\substack{i, j\in \N_0 \\ i+j=n}} \sum_{\sigma \in S_{n}} \frac{\chi(\sigma,v_1,\dots,v_n)}{i! j!} \\&\hspace{1cm}\times\left(v_{\sigma(1)}\wedge \ldots \wedge v_{\sigma\left(i\right)}\right)\boxtimes \left(v_{\sigma(i+1)}\wedge \ldots \wedge v_{\sigma\left(i+j\right)} \right)\\&=\sum_{\substack{i, j \in \mathbb{N}_0 \\ i+j=n}} \sum_{\sigma \in \operatorname{UnShuff}(i, j)} \chi\left(\sigma, v_1, \cdots, v_n\right) \\&\hspace{1cm}\times    \left(v_{\sigma(1)}\wedge\cdots\wedge v_{\sigma(i)}\right)\boxtimes\left( v_{\sigma(i+1)}\wedge \cdots\wedge v_{\sigma(n)}\right)
\end{split}
\end{equation}
where $\mathbb{N}_0=\N\cup\{0\}$ and $\boxtimes$ again denotes the external tensor product. 
The counit is same as for the tensor coalgebra \eqref{eq:counitTV}. It turns out that the vector space Coder($SV$) of coderivations of $SV$ is isomorphic to the vector space of linear maps Hom$(SV,V)$. The precise map is given by the following theorem.
\begin{thm}\label{thm:codtoend}
The following map 
\begin{equation}
\begin{split}
\mathrm{Hom}(SV,V)&\longrightarrow \mathrm{Coder}(SV),\\f&\longmapsto \mu\circ(f\boxtimes\1_{SV})\circ\Delta_{SV}~,
\end{split}    
\end{equation}
where $\mu:SV\boxtimes SV\longrightarrow SV$ is the multiplication map, is an isomorphism of vector spaces with the inverse given by
\begin{equation}
    \begin{split}
     \mathrm{Coder}(SV) &\longrightarrow \mathrm{Hom}(SV,V)  \\d &\longmapsto \pi_1\circ d~,    
    \end{split}
\end{equation}
where
\begin{equation}
\begin{split}
&\pi_n:SV\longrightarrow V^{\wedge n}.
\end{split}  
\end{equation}
is the projection of $SV$ to the 
$n^{\text{th}}$-homogenous component $V^{\wedge n}$.
\end{thm}
Linear maps in $\mathrm{Hom}(SV,V)$ are called symmetric multilinear maps. 
\label{app:a1}
The notion of a $g$-twisted $L_\infty$-superalgebra $(V,\{\ell_n\}_{n=1}^\infty,g)$ can be succintly expressed as the existence of a certain coderivation that is equivalent to the $n$-ary brackets on the symmetrized tensor coalgebra. 
Given a symmetric multilinear maps $f_i:V^{\wedge i}\longrightarrow V$, we can define an associated coderivation $\bm{f}:SV\longrightarrow SV$ using Theorem \ref{thm:codtoend}. First define the total linear map 
\begin{equation}
    f=\sum_{i=1}^\infty f_i~,
\end{equation}
where $f_i$ is nonvanishing only on $V^{\wedge i}$. Then the linear map $\bm{f}=\mu\circ(f\boxtimes \1_{SV})\circ\Delta_{SV}$ is the associated coderivation. One can decompose $\bm{f}$ into coderivations $\bm{f}_i$:
\begin{equation}
  \bm{f}=\sum_{i=1}^\infty \bm{f}_i~,  
\end{equation}
which vanishes on $V^{\wedge n}$ for $n<i$ and acts as 
\begin{equation}\label{eq:coderbmci}
    \bm{f}_i=\sum_{n\geq i}\left(f_i\wedge\mathds{1}_{V^{\wedge n-i}}\right)\pi_n~,
\end{equation}
on $SV$, see \eqref{eq:albewed} for the definition of $f_i\wedge\mathds{1}_{V^{\wedge n-i}}$.
\par
The sum of the $n$-ary brackets $\ell_n^g:V^{\wedge n}\longrightarrow V$ for\footnote{Recall that $\ell_1^g=\ell_1$.} $n\geq 1$ 
\begin{equation}
    \ell^g=\sum_{n=1}^\infty \ell_n^g~,
\end{equation}
gives us a linear map $\ell^g:SV\longrightarrow V$. Then by Theorem \ref{thm:codtoend}, we get a coderivation $\bm{\ell^g}:SV\longrightarrow SV$. 
Then the strong homotopy Jacobi identity is equivalent to the equation
\begin{equation}
(\bm{\ell^g})^2=\frac{1}{2}[\bm{\ell^g},\bm{\ell^g}]=0~,    
\end{equation}
where $[\cdot,\cdot]$ is the graded commutator defined by \eqref{eq:gradcommdef} (see Remark \ref{rem:lilj=1/2[]}).
\section{Homological perturbation lemma}\label{app:sdr}
In this section, we review the definitions of (strong) deformation retracts and the homological perturbation lemma. 
We follow \cite{crainic2004perturbation}, \cite[Appendix A]{Erbin:2020eyc} for this description. Let $V,W$ be two $\Z$-graded vector spaces. \begin{defn}\label{defn:SDR}
A homotopy equivalence (HE) is a pair\footnote{The general notion of homotopy equivalence is defined for arbitrary pair of complexes. We will consider the pair $(V,d_V)$ as a (co)chain complex with a single nontrivial (co)chain group.} $(V,d_V)$ and $(W,d_W)$ with degree-odd linear maps $d_V:V\longrightarrow V$ and $d_W:W\longrightarrow W$ along with two linear maps $\pi:V\longrightarrow W$ and $\iota:W\longrightarrow V$ satisfying the following conditions
\begin{enumerate}
    \item $d_V^2=d_W^2=0$ and $\pi,\iota$ are quasi-isomorphisms.\footnote{A quasi-isomorphism is a chain map between (co)chain complexes which descends to an isomorphism of (co)homology. Here $\pi,\iota$ are isomorphisms of the (co)homology on $V,W$ induced by $d_V,d_W$ respectively.}
    \item There exists a degree-odd linear map $\eta:V\longrightarrow V$ called the \textit{homotopy} map satisfying the \textit{Hodge-Kodaira decomposition}
    \begin{equation}\label{eq:hkdec}
        \iota\pi-\1_V=\eta d_V+d_V\eta.
    \end{equation}
\end{enumerate}
A \textit{deformation retract} (DR) is an HE with the requirement $\pi\iota=\1_W$. A \textit{strong deformation retract} (SDR) is a DR such that the homotopy map $\eta$ satisfies the \textit{annihilation conditions}
\begin{equation}\label{eq:anncond}
    \eta^2=\eta \iota=\pi \eta=0.
\end{equation}
\end{defn} 
We will represent an HE, DR or SDR by the compact diagram 
\begin{equation}\label{eq:sdrVW}
\rotatebox[origin=c]{270}{$\mathlarger{\mathlarger{\mathlarger{\mathlarger{\circlearrowright}}}}$}~\eta\left(V, d_V\right) \underset{\iota}{\stackrel{\pi}{\myrightleftarrows{\rule{2cm}{0cm}}}}\left(W, d_W\right)
\end{equation}
and specify whether it is an HE, DR or an SDR. 
\begin{remark}
The condition that $\pi,\iota$ be quasi-isomorphisms automatically require them to be (co)chain maps. In particular, they must satisfy
\begin{equation}\label{eq:chainmapprop}
    \begin{aligned}
d_W \pi & =\pi d_V, \\
\iota d_W & =d_V \iota.
\end{aligned}
\end{equation}
\end{remark}
\par
The homological perturbation lemma gives us a recipe to construct a new SDR from a given SDR. 
We are interested in the following version of the homological perturbation lemma.
\begin{thm}\label{thm:hpt}
Let 
\begin{equation}\label{eq:sdrVW}
\rotatebox[origin=c]{270}{$\mathlarger{\mathlarger{\mathlarger{\mathlarger{\circlearrowright}}}}$}~\eta\left(V, d_V\right) \underset{\iota}{\stackrel{\pi}{\myrightleftarrows{\rule{2cm}{0cm}}}}\left(W, d_W\right)
\end{equation}
be an HE (resp. DR, resp. SDR), $V',W'$ vector spaces and 
\begin{equation}
\begin{split}
&\delta_V:V\longrightarrow V'\otimes V,\quad g:V'\otimes V\longrightarrow V\\&\pi':V'\longrightarrow W',\quad g':W'\otimes W\longrightarrow W~,
\end{split}
\end{equation}
be linear maps satisfying the following conditions:
\begin{enumerate}
    \item $\Tilde{d}_V^2=(d_V+g\delta_V)^2=0$.
    \item The operator $(\1_V-\eta g \delta_V)$ is invertible.
    \item $g'(\pi'\otimes\pi)=\pi g$.
\end{enumerate}
Then the perturbed data 
\begin{equation}\label{eq:pertsdrVW}
\rotatebox[origin=c]{270}{$\mathlarger{\mathlarger{\mathlarger{\mathlarger{\circlearrowright}}}}$}~\tilde{\eta}\left(V, \tilde{d}_V\right) \underset{\tilde{\iota}}{\stackrel{\tilde{\pi}}{\myrightleftarrows{\rule{2cm}{0cm}}}}\left(W, \tilde{d}_W\right)~,
\end{equation}
is again an HE (resp. DR, resp. SDR) with the following definitions:
\begin{equation}
\begin{aligned}
\delta_W & =(\pi'\otimes\pi) \delta_V \frac{1}{\1_V-\eta g \delta_V} \iota, \\
\tilde{\iota}& =\iota+\eta g\delta_V \frac{1}{\1_V-\eta g\delta_V} \iota, \\
\tilde{\pi} & =\pi+\pi g\delta_V \frac{1}{\1_V-\eta g\delta_V} \eta, \\
\tilde{\eta} & =\eta+\eta g\delta_V \frac{1}{\1_V-\eta g\delta_V} \eta,
\end{aligned}
\end{equation}
and $\tilde{d}_W=d_W+\tilde{g}\delta_W$, where 
\begin{equation}
(\1_V-\eta g\delta_V)^{-1}=\frac{1}{\1_V-\eta g\delta_V}=\sum_{n=0}^\infty (\eta g\delta_V)^n.    
\end{equation}
The map $\delta_V$ is called a \emph{perturbation} of the HE (resp. DR, resp. SDR) \eqref{eq:sdrVW}.
\end{thm}
\begin{proof}
The proof is identical to \cite[Appendix A.3]{Erbin:2020eyc}, and \cite{crainic2004perturbation} with perturbation $g\delta_V$. 
\end{proof}
\end{appendix}
\bibliography{main.bib}
\end{document}